\documentclass[a4paper,UKenglish,cleveref, autoref, thm-restate]{lipics-v2021}
\usepackage{hyperref}
\graphicspath{{./figures/}}%helpful if your graphic files are in another directory
\title{Realizability of Free Spaces of Curves} %TODO Please add

\author{Hugo A. Akitaya} {Department of Computer Science, University of Massachusetts Lowell, USA}{hugo\_akitaya@uml.edu}{https://orcid.org/0000-0002-6827-2200}{}

\author{Maike Buchin}{Department of Computer Science, Ruhr University Bochum, Germany}{maike.buchin@rub.de}{https://orcid.org/0000-0002-3446-4343}{}

\author{Majid Mirzanezhad}{Transportation Research Institute, University of Michigan\,--\, Ann Arbor, USA}{miirza@umich.edu}{https://orcid.org/0000-0002-2950-673X}{}

\author{Leonie Ryvkin}{Department of Mathematics and Computer Science, Eindhoven University of Technology, the Netherlands}{leonie.ryvkin@tue.nl}{https://orcid.org/0000-0002-7036-1341}{}

\author{Carola Wenk}{Department of Computer Science, Tulane University, Louisiana, USA}{cwenk@tulane.edu}{https://orcid.org/0000-0001-9275-5336}{}

\funding{Carola Wenk is partially supported by NSF grant CCF 2107434}

\authorrunning{H. Akitaya and M. Buchin and M. Mirzanezhad and L. Ryvkin and C. Wenk}

\Copyright{Hugo Akitaya and Maike Buchin and Majid Mirzanezhad and Leonie Ryvkin and Carola Wenk}


 %TODO mandatory; please add comma-separated list of keywords

\category{} %optional, e.g. invited paper

\relatedversion{} %optional, e.g. full version hosted on arXiv, HAL, or other respository/website
\nolinenumbers %uncomment to disable line numbering

\usepackage[utf8]{inputenc}
\usepackage{amsmath}
\usepackage{amssymb}
\usepackage{amsfonts}
\usepackage{graphicx}
\usepackage{xspace}
\usepackage[ruled, vlined]{algorithm2e}
\usepackage{multirow}
\usepackage{thmtools}

\newcommand{\frechet}{Fr\'echet\xspace}
\newcommand{\fd}{\frechet distance\xspace}

\newcommand{\dfd}{discrete \frechet distance\xspace}

\newcommand{\fsd}{free space diagram\xspace}
\newcommand{\fsm}{free space matrix\xspace}

\newcommand{\nex}{\mathsf{Next}}
\newcommand{\pre}{\mathsf{Prev}}

\newcommand{\Reals}{\ensuremath{\mathbb{R}}}
\newcommand{\eps}{\ensuremath{\varepsilon}}

\newcommand{\ourremark}[2]{}

\newcommand{\majid}[1]{\ourremark{Majid}{#1}}

\newcommand{\todo}[1]{\ourremark{TODO}{#1}}

\newcommand{\highlight}[1]{\textcolor{blue}{#1}}

\newcommand{\figref}[1]{Figure~\ref{#1}}
\newcommand{\tabref}[1]{Table~\ref{#1}}
\newcommand{\secref}[1]{Section~\ref{#1}}
\newcommand{\lemref}[1]{Lemma~\ref{#1}}
\newcommand{\thmref}[1]{Theorem~\ref{#1}}
\newcommand{\obsref}[1]{Observation~\ref{#1}}
\newcommand{\corref}[1]{Corollary~\ref{#1}}

\begin{document}

 \maketitle

\begin{abstract}
The free space diagram is a popular tool to %analyze and 
compute the well-known Fréchet distance. 
As the Fréchet distance is used in many different fields, %such as computer-aided design, geographic data analysis and comparison of protein chains, 
many variants have been established to cover the specific needs of these applications. 
Often the question arises whether a certain pattern in the free space diagram is \emph{realizable}, %or not, 
i.e., whether there exists a pair of polygonal chains whose free space diagram corresponds to it. %the pattern in question.
The answer to this question may help in deciding the computational complexity of these distance measures, %whose algorithms make use of free space diagrams, 
as well as allowing to design more efficient algorithms for restricted input classes that avoid certain %undesired or time consuming 
free space patterns.
Therefore we study the inverse problem: Given a potential \fsd, do there exist curves that generate this diagram? \majid{Have we not already defined the problem above?}

Our problem of interest is closely tied to the classic Distance Geometry problem. We settle the complexity of Distance Geometry in $\Reals^{> 2}$, showing $\exists\Reals$-hardness.
We use this to show that for curves in $\Reals^{\ge 2}$ the realizability problem is $\exists\Reals$-complete, both for continuous and for discrete Fr\'echet distance. 
%
%For the continuous case in $\Reals^2$ we provide polynomial and exponential time algorithms for special cases. \leo{Remove the last sentence and cite EuroCG papers?}
%
We prove that the continuous case in $\Reals^1$ is only weakly NP-hard, and we provide a pseudo-polynomial time algorithm and show that it is fixed-parameter tractable.
Interestingly, for the discrete case in $\Reals^1$ we show that the problem becomes solvable in polynomial time.
\end{abstract}

\setcounter{page}{1}

\section{Introduction}
\label{sec:intro}

%\hugo{Draw parallels with other geometric problems such as recognizing visibility graphs, distance geometry. This book \cite{noggle2012nuclear} talks a bit about heteronuclear nuclear Overhauser effects (NOE) in Nuclear magnetic resonance (NMR) spectroscopy which is used to infer distances between a proton (H) and a different nucleus of different type (such as C, H, or F). Given that such distances are prone to noise and can only be determined when the distances are small, this fits our model. Paper that failed proving that diatnace geometry is in NP~\cite{beeker2012distance}. }

The \fd is arguably the most popular % an important 
distance measure for curves in computational geometry and has been studied extensively in the last years. It has application in various fields, including geographic data analysis and the comparison of protein chains~\cite{eghmm-nsmbp-02, glw-mpstd-07, skb-fdbas-07, jxz-pssad-08}.  
For the latter application, typically the well-established variant, %of the classical \fd,
the \dfd, is used. %plays an important role. 
The standard tool for computing the \fd of two curves is the {\em \fsd}, which is the cross-product of the parameter spaces of the curves partitioned into {\em free space} and its complement, where free space is the sublevel set of the distance function for a given $\eps>0$. 
For two piecewise linear curves of $m$ and $n$ line segments parameterized by their natural arc-length parametrization, it is well-known that the \fsd\ consists of $mn$ cells, and the free space in each cell has the shape of a cropped ellipse \cite{altgodau}. 
The \fd is at most $\eps$ if and only if there exists a monotone path in the \fsd\ that covers the parameter spaces of both curves.
Hence, to compute the \fd one searches for such a path in the \fsd. 

For different applications, many variants of the \fd have been developed, which are typically also computed using the \fsd.
%, which needs to be analyzed for this. 
The \dfd relies on (a discretization of the \fsd) the {\em \fsm}: For two discretized curves given by $n$ and $m$ points, respectively, the $n\times m$ free space matrix with entries $a_{i,j}\in\{0,1\}$ captures whether two points $i,j$ of different curves lie within $\eps>0$ distance of one another, or not. 
The \dfd of two curves is at most $\eps$ iff there exists a monotone ``path'' of 1 entries connecting opposing corners in the \fsm.

Runtimes of the resulting algorithms usually depend on the complexity of the \fsd or \fsm~\cite{dhw-afdrc-12, ahkww-fdcr-06, adaptiveDisFrechet-18}. It is known that neither \fd nor \dfd can be computed in subquadratic time unless SETH fails~\cite{bringmann, fastweak, bw-adfd-15}. 
However, there are several faster algorithms %for the computing the \fsd\ 
for special curve classes such as~\cite{knauer,driemel}, which exploit a special structure of the free space diagram. %\majid{I would really mention more special cases and their implicit structures if it's all about the structure of free space like packing arguments, reducing reachable space and such...}
%\carola{I'll add some more detail here and cite some more papers, for continuous.}
%
%\leo{Do we know sth about the complexity for \fsm? Some special variants there?}
%\carola{I don't know anything about discrete Frechet though.}
%
These complexity bounds always consider the worst-case complexity of the \fsd, but not every \fsd\ can be realized by a pair of curves. 
Also, some variants of the \fd have proven to be NP-hard to decide~\cite{isaac,shortcut2}, some of which build certain free space diagrams for the reduction. 

Here we therefore study the inverse problem: Given a (potential) \fsd  (\fsm), do there exist curves (ordered point sets) that generate this \fsd or matrix? 
To our knowledge this {\em free space realizability problem} has so far only been studied for special cases~\cite{phdleo, eurocg21, eurocg22}. %not yet been studied before. 
Understanding it %this inverse problem 
will give structural insights into free spaces and the computation of the \fd, in particular for special curve classes.  %\maike{mention EuroCG papers or Leos thesis here?}
%\maike{first overview or related problems?}
Note that although we gained a good understanding of the realizability problem in the settings described below, other settings remain to be investigated further.
We are particularly interested in studying settings where less information is given in the free space diagram or matrix, e.g. only some cells or only cell boundaries are provided with the input.
\subparagraph*{Overview.}
%In \secref{sec:realdef} we recall necessary definitions and formally define the realizability problem. % that we study. 
%and specify the input instances we consider. 
%
We give results (see \tabref{tab:resultsOverview}) %for different variants of the realizability problem: 
for both the continuous and the discrete variant of the problem, with \fsd and \fsm inputs, %as well as the discrete variant with \fsm inputs, 
and curves may be realized in $\Reals^d$, with $d=1$ or $d\geq 2$.
%
%Our results are summarized in \tabref{tab:resultsOverview}. 
%
We show that for curves in $\Reals^{\ge 2}$ the realizability problem is $\exists\Reals$-complete both for the continuous case (\secref{sec:continuous2DHardness}) and for the discrete case (\secref{sec:discrete2DHardness}).
For the continuous case in $\Reals^2$, algorithms are known only for special cases~\cite{phdleo, eurocg21}. %\cite{eurocg21}
%Nevertheless, for the continuous case in $\Reals^2$ we are able to provide algorithms for special cases (\secref{sec:continuous2DAlgos}).
%
%We then study realizability in $\Reals^1$.
For curves in $\Reals^1$, the problem in the discrete case interestingly becomes solvable in polynomial time (\secref{sec:discrete1Dpoly}), while in the continuous case, it is weakly NP-hard (\secref{sec:continuous1D}) and fixed-parameter tractable, %(\secref{sec:continuous1D}),
%where the parameter $k$ is the number of empty or full columns in the free space, 
and we also provide a pseudo-polynomial time algorithm %(\secref{sec:continuous1Dpseudopoly}).
(\secref{sec:continuous1D}).

\begin{table}[htbp]
    \centering
    \begin{tabular}{|p{3cm}|l|l|}
    \hline
    \small
     Input & $\Reals^d$   & Results  \\\hline\hline
     \multirow{5}{*}{\begin{minipage}{3cm} Free Space Diagram \end{minipage}} & $d\ge 2$ &  $\exists\Reals$-complete (\highlight{\thmref{thm:er-hardness}}) \textbf{\secref{sec:continuous2DHardness}}\\
     && Algorithms for special cases %(\highlight{Theorems \ref{thm:realpolyalgo} and \ref{thm:realexpalgo}}) \textbf{\secref{sec:continuous2DAlgos}}
     (\cite{phdleo, eurocg21})\\\cline{2-3}%\cite{eurocg21}
     & & weakly NP-hard (\highlight{\thmref{thm:1D-continuous-hardness}}) \textbf{\secref{sec:continuous1D}}\\
     & $d=1$ & FPT $O(mn2^k)$ (\highlight{\thmref{thm:alg-fpt-continuous1d}}) \textbf{\secref{sec:continuous1D}}\\
     && Pseudo-poly time (\highlight{\thmref{thm:pseudo-poly}}) \textbf{\secref{sec:continuous1D}}\\\hline
     \multirow{2}{*}{\begin{minipage}{3cm} Free Space Matrix \end{minipage}} & $d\ge 2$ & $\exists\Reals$-complete (\highlight{\thmref{thm:ER-discrete}}) \textbf{\secref{sec:discrete2DHardness}}\\\cline{2-3}
     & $d=1$ & %$O\big(nm+\min(m^2,k^2n)\big)$ 
     $O(nm^2)$
     (\highlight{\thmref{thm:uig-runtime}}) \textbf{\secref{sec:discrete1Dpoly}}\\\hline        
    \end{tabular}
    \caption{Overview of our and known results. \majid{any explanation for $k$? \# of empty strips?}}
    \label{tab:resultsOverview}
\end{table}

\subparagraph*{Related problems.}
The exploration of inverse problems is a recurrent subject in computational geometry, often applied to
%with frequent applications to 
recognition and reconstruction problems. Notable examples 
%of such inverse problems include 
are the \emph{inverse Voronoi Diagram} problem~\cite{ash1985recognizing} and the \emph{visibility graph recognition} problem~\cite{boomari2018recognizing}. 
Our problem of interest can be viewed as a curve embeddability problem given certain proximity criteria on edge lengths and point-to-point distances between two curves in their free space diagram. %Such criteria relate our problem to classic graph embeddability problems.
Our results have a significant impact on problems involving distance constraints on geometric graphs, such as \emph{Distance Geometry} and \emph{Disk Intersection Graphs}. 

%\begin{itemize}
    %\item 
{\bf Distance Geometry.}
Our results closely relate to problems involving distance constraints on geometric graphs, such as the %distance geometry problem and disk/sphere intersection graphs. 
\emph{Distance Geometry} problem, which is a classic inverse problem in computational geometry. It involves embedding an abstract weighted graph in $\Reals^d$ Euclidean space such that the Euclidean length of each edge corresponds to its weight~\cite{saxe1979embeddability}. This problem is equivalent to \emph{Linkage Realizability}, where the edges correspond to rigid bars in a mechanical linkage~\cite{saxe1979embeddability}. %The problem has applications in molecular conformation, where the shape of a molecule is inferred from distances between atoms obtained via NMR (nuclear magnetic resonance) spectroscopy~\cite{saxe1979embeddability}. 
While distance geometry was shown to be NP-hard in the 70s, its membership in NP remained open until Schaefer showed $\exists\Reals$-completeness in 2012~\cite{schaefer2012realizability} for $\Reals^2$. Whether the problem remains $\exists\Reals$-hard in higher dimensions was posed as an open problem by Schaefer. % and remained an open question.
%\textcolor{red}{
%Schaefer posed as an open question whether the problem remains $\exists\Reals$-hard in higher dimensions.
To show that the continuous version of the free space realizability  
problem in $\Reals^{\ge 2}$ is $\exists\Reals$-hard we reduce from distance geometry, using a gadget from~\cite{saxe1979embeddability}. %to show that distance geometry is indeed $\exists\Reals$-hard in $\Reals^{\ge 2}$.
%To the best of the authors' knowledge, although the ingredients of this proof were already known, we are the first to claim the result.
%}

%\item 
{\bf Unit Sphere Graphs.}
The sphericity of a graph is the minimum dimension for which the graph has a unit sphere representation. Unit sphere graph realizability is known to be $\exists\Reals$-hard. Havel first introduced the study of sphericity in the context of molecular conformation~\cite{havel1982combinatorial}. A unit sphere graph is an intersection graph of unit spheres and can be seen as a complete graph where each edge is marked with a distance constraint $\le 1$ or $>1$. The problem of realizing a free space matrix corresponds to realizing a complete bipartite graph where each edge is marked with a distance constraint $\le 1$ or $>1$ (Section~\ref{sec:discrete2DHardness}). This defines a class of graphs, as do unit disk graphs. In contrast to a unit disk graph, there are pairs of vertices (the ones in the same partite set) whose distances are dispensable. %Unit disk graphs admit arbitrary clique sizes while not admitting $K_{1,6}$ as an induced subgraph, for example. On the other hand, this bipartite distance-constrained graph admits arbitrary bicliques while being triangle-free.
A similar class are visibility graphs, the recognition of which is also known to be $\exists\Reals$-complete~\cite{ch-rcpvg-17}. %\leo{Should we leave this sentence out? Was requested by first reviewer...}
%\maike{I would leave it in for the reviewer}

%\item 
{\bf Bipartite Distance-Constrained Graphs.}
Modeling distance constraints in general (non-complete) graphs is useful when data is unavailable between every pair of nodes or due to the topology of the underlying network. E.g., heteronuclear NMR is used to obtain less cluttered and less noisy data~\cite{noggle2012nuclear}. This allows the inference of distances between two different types of atoms, and thus, the distances constraints form a bipartite graph. 
%\end{itemize}

% Our results are then given in subsequent sections: 
% \noindent\textbf{Continuous, 2D}
% \begin{itemize}
%     %\item max number of realizing cells is NP-hard 
%     \item ER-complete % Hugo (need some polishing and make the full proof.)
%     (\secref{sec:continuous2DHardness})
%     \item polytime \& exptime algorithm for special cases %  Maike & Leo 
%     (\secref{sec:continuous2DAlgos})
% \end{itemize}

% \noindent\textbf{Continuous, 1D}
% \begin{itemize}
%     \item weakly NP-hard %(basically done, needs shortening; Leo)
%     (\secref{sec:continuous1DHardness})
%     \item Pseudo-poly-time 1D %(not written yet): Hugo (I'll prioritize this) & ...
%     (\secref{sec:continuous1Dpseudopoly})
%     \item FPT algorithm (using folding) % (basically done, needs shortening)
%     (\secref{sec:continuous1DFPT})
% \end{itemize}

% \noindent\textbf{Discrete, 2D}
% \begin{itemize}
%     \item ER-hardness 2D  %Hugo (I'll do this next)
%     (\secref{sec:discrete2DHardness})
% \end{itemize}

% \noindent\textbf{Discrete, 1D}
% \begin{itemize}
%     \item 1D poly-time algorithm  %Majid
%     (\secref{sec:discrete1Dpoly})
% \end{itemize}

%\input{preliminaries}
\section{Preliminaries}\label{sec:realdef}
\subparagraph*{Continuous case.}
Let $P=(p_0,\ldots,p_n)$ and $Q=(q_0,\ldots,q_m)$ be polygonal curves in $\Reals^d$ of lengths $\ell_P$ and $\ell_Q$, continuously parameterized by arc-length, i.e. $p_i=P(i)$ and $q_j=Q(j)$ for $i \in [n]$, $j \in [m]$, where $[n]= \{1, \cdots, n\}$.
%We briefly recall that for (parameterized) polygonal curves $P=(p_0,\ldots,p_n)$ and $Q=(q_0,\ldots,q_m)$ in $\Reals^d$ of lengths $\ell_P$ and $\ell_Q$, the 
Given $\eps>0$, their \emph{free space} is defined as $F_\eps(P,Q) = \{(r,t) \mid \Vert P(r)-Q(t)\Vert \leq \eps\}$. The \emph{free space diagram} puts this information in an $m \times n$ grid: We define $D_\eps(P,Q)$ as the colored rectangle $R=[0, \ell_P]\times [0, \ell_Q]\subseteq \Reals^2$, where a point $(p,q)\in R$ is colored white iff $(p,q)\in F_\eps(P,Q)$.
The grid $X$, which is the set of segments $\{p_i\times[0,\ell_Q] \mid i\in \{0, \dots, n\}\} \cup \{[0,\ell_P]\times q_j \mid j\in \{0, \dots, m\}\}$, subdivides $R$ into $n \times m$ cells $C_{i,j}$.
%Note that we distinguish between \emph{components} $c \subseteq F_\eps(P,Q)$, which are connected subsets of free space, and \emph{cell-components} $c_{i,j}=C_{i,j}\cap F_\eps{P,Q}$.
%We define further important terms concerning the free space diagram:
We call a single cell $C_{i,j}$ \emph{empty} %(or \emph{gray}) 
if $C_{i,j} \cap F_\eps = \emptyset$ and \emph{full} %(or \emph{white}) 
if $C_{i,j} \cap F_\eps = C_{i,j}$. %\maike{do we actually use the terms gray and white?}
If $\emptyset \neq C_{i,j} \cap F_\eps \neq C_{i,j}$, the cell $C_{i,j}$ is called \emph{partially full}. 
A diagram $D_\eps$ is called \emph{realizable} if there exist curves $P, Q$ such that $D_\eps(P,Q)=D_\eps$.
%Note that free space (diagram) is a term defined by the corresponding curves.
%Thus we call the considered diagrams \emph{input} or \emph{given diagrams} 
%Given a diagram $D_\eps$ we ask whether there exist curves $P$ and $Q$ such that $D_\eps=D_\eps(P,Q)$.
%, moreover, \emph{white space} in $D_\eps$ equals free space in $D_\eps(P,Q)$.
%
%Now, for a given diagram $D_\eps$ we ask whether it is realizable. %there exist curves $P,Q$ such that $D_\eps=D_\eps(P,Q)$.
We assume that the exact lengths of all cell boundaries and the equation of each component's boundary curve %and the value~$\eps$ 
are part of the input. %. In this regard, 
We consider the real RAM computation model throughout the paper.
In \figref{fig:realintroexample}, the leftmost diagram is not realizable:
%the curves $P$ and $Q$ obtained from the components in cells $C_{1,1}$ and $C_{2,2}$ induce the free space diagram on the right-hand side, which features an additional component in cell $C_{1,2}$.
Upon fixing the placement of segments corresponding to cell $C_{1,1}$, we have two options to place the remaining segments such that cell $C_{2,2}$ is realized.
This either induces components in none of the remaining cells, or in both.

\begin{figure}[ht]
	\centering
	\includegraphics[scale=0.6]{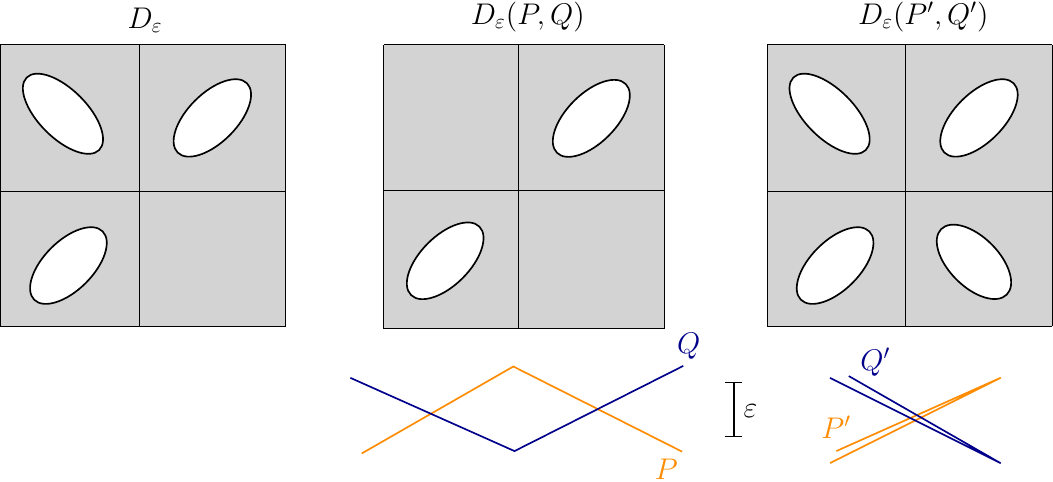}
	%	\caption{A given diagram $D_\eps$, the curves as computed from the partially full cells, and their free space diagram $D_\eps(P,Q)$.}
	\caption{Given diagram $D_\eps$, there are two ways to place curves $P,Q$ in $\Reals^2$. Neither realizes $D_\eps$.}
	\label{fig:realintroexample}
\end{figure}

In the following, we denote the %sphere of radius $r$ centered at a point $x$ as $\mathcal{C}_r(x) :=\{p \in \mathbb{R}^2 \mid \Vert x-p\Vert = r\}$ and the corresponding 
ball of radius $r$ centered at a point $x
\in \Reals^d$ as $\mathcal{B}_r(x) :=\{p \in \mathbb{R}^d \mid \Vert x-p\Vert \leq r\}$.
The \emph{$\eps$-neighborhood} of an object $X$ is given by $\bigcup_{x \in X}\mathcal{B}_\eps(x)$. 
We denote by $s_i^P=\overline{p_{i-1}p_i}$, $i\in [n]$, the line segment connecting consecutive vertices of $P$, and define $s_i^Q$ analogously.
%
%\paragraph*{Specific notation for curves in 1D}
%In 1D we also specify a polygonal curve $P\colon I \to \Reals$ by vertices $p_0, \ldots, p_n\in\Reals$, and call the line segment connecting consecutive vertices $s_i^P=\overline{p_{i-1}p_i}$. %\maike{don't we use this notation in 2D as well?}
In $\Reals^1$, 
if two consecutive segments $s^P_i, s^P_{i+1}$ have different orientations (segments are placed on top of each other), we say the curve \emph{folds} at the common \emph{folding vertex} $p_i$.
Else, we say that the curve is \emph{straight} at $p_i$.
It is known that the free space of two lines has the shape of a cropped ellipse with axis at $\pm 45^\circ$~\cite{altgodau,Rote}. For curves in $\Reals^1$, the lines are necessarily parallel, hence the ellipse degenerates to a slab bounded by lines at $\pm 45^\circ$~\cite{Buchin2017}.
%(infinitely wide major axis).  

%The definitions for \emph{empty}, \emph{full}, and \emph{partially full} cells apply as given above.
%An empty cell corresponds to a pair of non-overlapping segments whose endpoints have pairwise distances $>\eps$; a full cell stems from segments where any interval of length $\eps$ centered at an arbitrary point of either segment fully contains the other segment.
%Partially full cells correspond to segments where at least one endpoint has distance $<\eps$ to some point of the other segment. 
%The free space $F_{\eps}(s_i^P,s_j^Q)$ in a cell is a slab, i.e., the space between two diagonal lines tilted by $\pm 45^\circ$, cropped at cell boundaries, see \figref{fig:partiallycell}.
%The horizontal and vertical lines bounding cells $C_{ij}$ in the diagram are denoted as \emph{grid lines}.
%
%Note that free space (diagram) is a term defined by the corresponding curves.
%Thus we call the considered diagrams \emph{input} or \emph{given diagrams} $D_\eps$ and ask whether there exist curves $P$ and $Q$ such that $D_\eps=D_\eps(P,Q)$, moreover \emph{white space} in $D_\eps$ equals free space in $D_\eps(P,Q)$.

\subparagraph*{Partially full cells.}
We now establish a necessary condition for curves $P, Q$ to realize $D_\eps$ that is used in Sections~\ref{sec:continuous2DHardness}--\ref{sec:continuous1D}.
Partially full cells give us information about the relative placement of the segments. 
We use the term \emph{relative placement} of two segments to mean that we know the distances between the intersection point of the lines containing the segments and the segment endpoints. %\footnote{Note that from this we can also deduce the angle between the lines containing them and the distances of the segments' start points to the lines' point of intersection.}. 
Note that after fixing the position of one segment, this still allows for two symmetric placements of the second segment. 

\begin{restatable}{lemma}{realpartially}\label{lem:realpartiallyfull}
	Given a partially full free space cell, $\eps>0$, and four points on the boundary of the ellipse in the cell, none of which are mirror images of another with respect to the ellipse's major and minor axes, we can compute the corresponding segments' relative placement.
\end{restatable}

The proof %(in Appendix~\ref{app:2Dlemma}) 
(see Lemma~5.5 in~\cite{phdleo}) relies on knowing that a cell is an ellipse at $45^\circ$. %~\cite{Rote}. 
Note that we obtain much less information from full or empty cells, namely only that the segments do or do not lie within or not within distance $\eps$ from each other. 

\subparagraph*{Definitions: Discrete case.}
For discrete polygonal curves (i.e., point sequences) $P,Q$ with $n$ and $m$ points, resp., the free space is defined as $F_\eps(P,Q) = \{(i,j) \in [n] \times [m] \mid \Vert p_i- q_j\Vert \leq \eps\}$.
We define the \fsm $M_\eps(P,Q)$ as the $n \times m$ matrix featuring entries $a_{i,j}\in \{0,1\}$, $i \in [n]$, $j \in [m]$ where $a_{i,j}=1$ if and only if $(i,j)\in F_\eps(P,Q)$.
%
%Similar to the continuous case, for 
For a given matrix $M_\eps$ 
%with entries being either one or zero, 
we ask whether there exist curves $P,Q$ such that $M_\eps=M_\eps(P,Q)$.
%In \figref{fig:realintroexample}, the free space diagram on the left-hand side is not realizable:
%the curves $P$ and $Q$ obtained from the two components in cells $C_{1,2}$ and $C_{2,1}$ correspond to the free space diagram on the right-hand side of the figure, which features an additional component in cell $C_{1,1}$.

%\input{ER-continuous.tex} %3
\section{$\exists\Reals$-Completeness for Continuous Curves\label{sec:continuous2DHardness}}
%\label{sec:ER-continuous}

%\vskip -1mm
We first show that given a diagram $D_\eps$, the problem of finding two curves in $\Reals^2$ that realize $D_\eps$ is $\exists\Reals$-complete. We then generalize to higher dimensions in Section~\ref{sec:highDimContinuous}.
Containment in $\exists\Reals$ is shown by expressing the problem using real inequalities, see Lemma~\ref{lem:ER-containment-continuous}, Appendix~\ref{app:ER-cont2D}. 

We reduce from the problem of deciding whether a linkage has a planar realization which was shown $\exists\Reals$-hard by Abel et al.~\cite{abel2016needs,abel2016folding}.
A \emph{mechanical linkage} is a mechanism made of rigid bars connected at hinges.
The input is a weighted graph $G=(V(G),E(G), \ell_G)$, where  $\ell_G$ is the weight function, and a function $\Pi\colon W\rightarrow\Reals^2$, where $W\subseteq V(G)$, that represents vertices whose positions are \emph{pinned}.
A \emph{configuration} $C$ of a linkage $\mathcal{L}=(G,\Pi)$ is a straight-line drawing of $G$ where the length of each edge $e\in E(G)$ is $\ell_G(e)$ and the position of each vertex $w\in W$ is $\Pi(w)$. 
The linkage realization problem asks whether a given linkage admits a configuration.
A configuration $C$ is \emph{noncrossing} if $C$ is a plane drawing.
%
%\subsection{Constrained Realization of Linkages}
%\label{sec:link}
Abel et al.~\cite{abel2016needs,abel2016folding} showed that the linkage realization problem remains hard for a series of restrictions on the input linkage $\mathcal{L}$.
We restate a direct consequence of Theorems~2.2.13 and~2.4.6 in \cite{abel2016folding}. %summarizing the restrictions.
Although not all conditions in the theorem below are explicitly stated in \cite[Theorem 2.2.13]{abel2016folding}, they can be directly inferred by their construction in \cite[Section 2.7]{abel2016folding}.

\begin{theorem}
	\label{thm:abel}[Simplified from 
	%Theorem 2.2.13 in 
	\cite{abel2016folding}, Theorems 2.2.13 and 2.4.6]
	Given a linkage $\mathcal{L}=(G,\Pi)$ and a combinatorial embedding (clockwise circular order of edges around each vertex) $\sigma$ of $G$, deciding whether there exists a planar realization of $\mathcal{L}$ is $\exists\Reals$-hard even if the following constraints are enforced:
	\begin{enumerate}
		\item \label{res:int} $G$ is connected, and the length of every edge is an integer.
		\item \label{res:rigid} A set of edge disjoint subgraphs $H$ of $G$ can be assigned \emph{rigid}, i.e., each angle between consecutive %(in $\sigma$) 
		incident edges in $H$ is prescribed from $\{90^\circ, 180^\circ,270^\circ, 360^\circ\}$. 
		Each subgraph $H$ is a tree, and an edge in $E(G)\setminus E(H)$ incident to $H$ must be incident to a leaf of $H$.
		%Vertices that are incident to at most one edge of each rigid subgraph are called \emph{nonrigid}.
		\item \label{res:pin}Only three vertices are pinned ($|\Pi|=3$), all three belong to the same rigid subgraph $H$ (described in constraint (\ref{res:rigid})), and they are not collinear.
		\item For every noncrossing configuration $C$ of $\mathcal{L}$ that satisfies constraints (\ref{res:int}--\ref{res:pin}), holds:
		%we are promised that:
		\begin{enumerate}
			\item $C$ agrees with $\sigma$.
			\item \label{res:angle} Angles that are not prescribed by constraint (\ref{res:rigid}) lie strictly between $60^\circ$ and $240^\circ$.
			\item The minimum distance of a vertex and a nonincident edge is at least a constant $\phi$.%feature size
		\end{enumerate}
	\end{enumerate}
\end{theorem}

We call a vertex \emph{rigid} if it is incident to at least two edges of the same rigid subgraph $H$, and \emph{nonrigid} otherwise.
By (\ref{res:rigid}), every angle incident to a rigid vertex is prescribed while no angle in a nonrigid vertex is prescribed (which by (\ref{res:angle}) can only vary in the interval $(60^\circ,240^\circ)$).
Note that distance geometry is equivalent to linkage realization with $\Pi=\emptyset$. Since (\ref{res:pin}) makes $\Pi$ irrelevant, Theorem~\ref{thm:abel} also implies hardness for distance geometry.

\smallskip\noindent\textbf{Reduction description.}
%\maike{Hugo, one reviewer complained that they did not understand the reduction. How much have you modified the description?} \hugo{I added the new figure of the whole reduction and tweaked the description a bit.}\maike{great!}
Given $\mathcal{L}$ and  $\sigma$ satisfying the constraints in \thmref{thm:abel}, we construct an instance $D_\eps$ as follows. 
A full example can be seen in Figure~\ref{fig:full-example-f-s}. %in Appendix~\ref{app:ER-cont2D}. 
The idea is to build a free space such that realizing curves trace out the linkage, following the given combinatorial embedding. For this, we first transform $G$ into a tree. The angle constraints in the linkage can also be enforced in the free space using a specific gadget. 

\begin{figure}[ht]
	\centering
	\includegraphics[scale=.6]{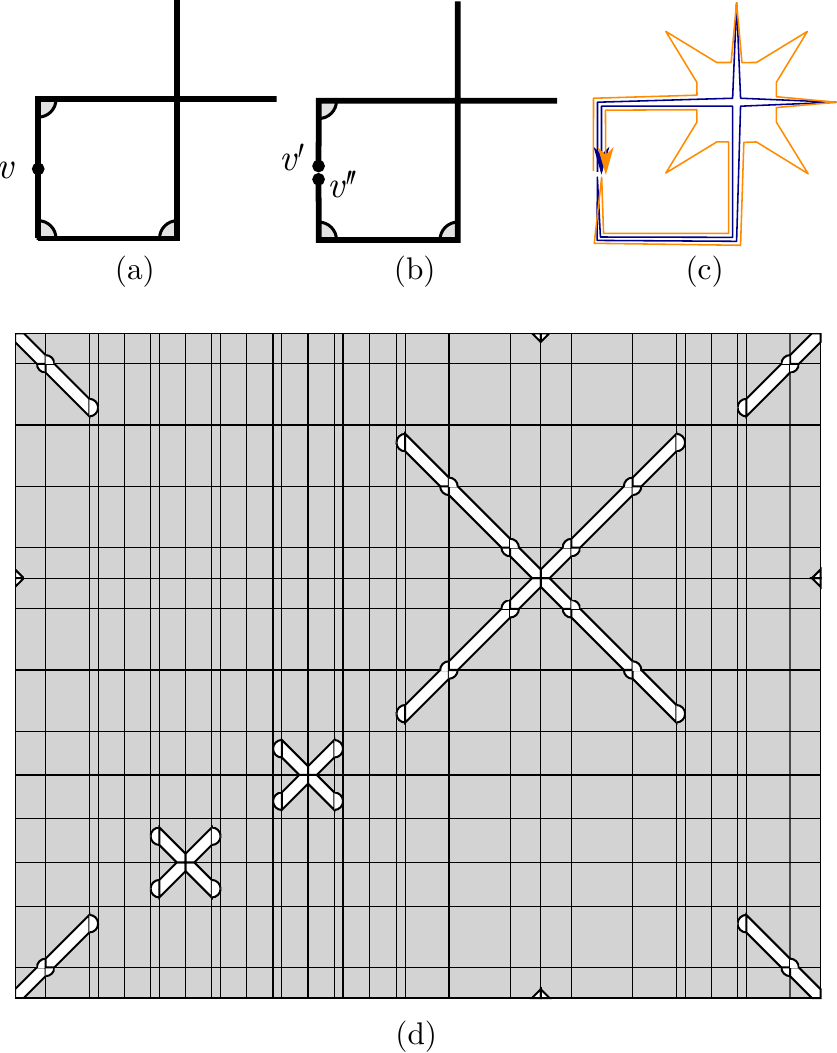}
	\caption{Example of our reduction from linkage realizability to \fsd realizability. (a) An input linkage $\mathcal{L}=(G,\Pi)$ and a subdivision vertex $v$ in a cycle of $G$. Rigid vertices are marked with gray angles. (b) Splitting $v$ transforms $G$ into a tree $T$. (c) The curves $P$ and $Q$ in $\Reals^2$ obtained from $T$. (d) The obtained \fsd.}
	\label{fig:full-example-f-s}
\end{figure}

While there is a cycle in $G$, split one edge in a cycle by placing a new vertex in its midpoint and performing a vertex split, creating two copies of the new vertex, each attached to half of the original edge.
We end up with a tree $T$.
Let $T'$ be the multigraph obtained by doubling each edge of $T$.
Intuitively, $D_\eps$ forces the curves $P$ and $Q$ to roughly trace a planar Eulerian circuit of $T'$ using the combinatorial embedding $\sigma$. 
(Up to a reflection and translation since $D_\eps$ can only specify the relative placement of $P$ and $Q$.)
%More precisely, 
$Q$ is exactly a planar Eulerian circuit of $T'$ while $P$ %will 
traces the same circuit but avoids an $\eps$-neighborhood of each nonrigid vertex using our \emph{angle gadget} (described later), which allows these angles to lie freely between $60^\circ$ and $240^\circ$.
%There are two possible ways to choose a planar Eulerian circuit of $T'$ agreeing with $\sigma$: circling the ``outline'' of the embedding $\sigma$ clockwise or counterclockwise. 
%\maike{not so important -- because it is clear that a Eulerian circuit exists -- but ``circling the outline'' would only be clear to me if $G'$ is a tree}
Both $P$ and $Q$ trace the ``outline'' $\sigma$ \emph{counterclockwise}. %, i.e., if we join the endpoints in $P$ or $Q$ we obtain a cycle that traces counterclockwise the perimeter of a weakly simple polygon. 
W.l.o.g.~assume $\phi\ge 6$, scaling the linkage by a constant factor if necessary.
We chose $\eps=1$ so that edges of $P$ and $Q$ that correspond to an edge $e$ of $G$ are close to each other and far from other edges.
Since every partially full cell determines the relative position of the corresponding pair of edges,
the four edges (two from $P$ and two from $Q$) that correspond to the traversal of $e$ are fixed relative to one another and lie on top of each other. 
They then simulate edge $e$.
The angle gadget guarantees flexibility so that the angle between incident edges can vary accordingly.
We add free space components to make the newly introduced subdivision vertices rigid: Their relative position is locked by \lemref{lem:realpartiallyfull} forming a $180^\circ$ angle, see the four small components on the sides of the \fsd in Figure~\ref{fig:full-example-f-s}d.

\begin{figure}
	\centering
	\includegraphics[width=\textwidth]{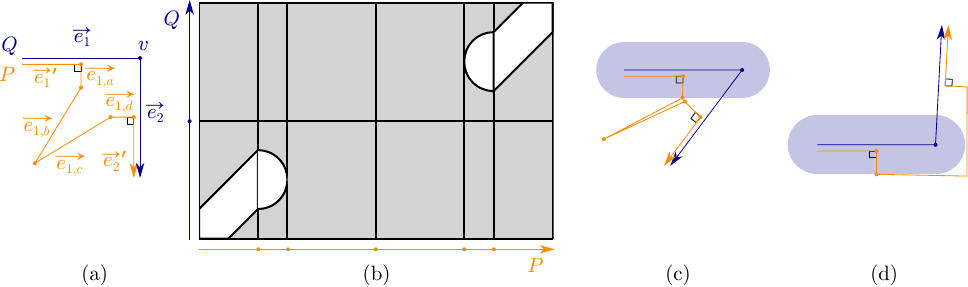}
	\caption{The angle gadget. (a) The $90^\circ$ configuration and (b) its free space diagram. (c) and (d) show the extremal configurations of the gadget with angles $2\cdot\tan^{-1}(1/2)\approx 53.13^\circ$ and $270^\circ$, resp.}
	\label{fig:angle-gadget}
\end{figure}

%We now describe the \emph{angle gadget}, shown in \figref{fig:angle-gadget}.
The \emph{angle gadget}, see \figref{fig:angle-gadget}, is represented by the 12 free space cells shown in \figref{fig:angle-gadget}(b). 
%\textsf{AngGadget}$(v)$ \majid{do we need this notation?} 
It is located at a small neighborhood of a vertex $v$ of $Q$; the figure only shows the portion of the free space relative to this neighborhood.
%We refer to this angle gadget as \textsf{AngGadget}$(v)$.
Note that $v$ is a degree-2 copy of a vertex $v^*$ of $G$. For clarity, we refer to all the copies of $v^*$ in $Q$ with different labels (by construction, there are $\deg(v^*)$ copies of each $v^*\in V(G)$, except for the starting vertex of the Eulerian circuit, which has an extra copy).
Let $\overrightarrow{e_1}$ and $\overrightarrow{e_2}$ be the two edges of $Q$ incident to $v$, and 
let $\overleftarrow{e_1}$ and $\overleftarrow{e_2}$ be the corresponding copies going in the opposite direction in $Q$, respectively.
Locally, $P$ has two edges $\overrightarrow{e_1}'$ and $\overrightarrow{e_2}'$ that overlap with $\overrightarrow{e_1}$ and $\overrightarrow{e_2}$, respectively.
We enforce the overlap by making all free space cells relative to $\overrightarrow{e_1}'$ (resp., $\overrightarrow{e_2}'$) empty except for the ones relative to $\overrightarrow{e_1}$ and $\overleftarrow{e_1}$ (resp., $\overrightarrow{e_2}$ and $\overleftarrow{e_2}$) which are partially full, containing an upward and downward $45^\circ$ full strip.
The distance between $v$ and the endpoints of $\overrightarrow{e_1}'$ and $\overrightarrow{e_2}'$ closest to $v$ is $2$ by \lemref{lem:realpartiallyfull}.
We place four edges $(\overrightarrow{e_{1,a}}, \overrightarrow{e_{1,b}},\overrightarrow{e_{1,c}},\overrightarrow{e_{1,d}})$ between $\overrightarrow{e_1}'$ and $\overrightarrow{e_2}'$ of lengths 1, 3, 3, and 1 in this order.
Only edges of length 1 have corresponding partially full cells: 
$C_{\overrightarrow{e_{1,a}}, \overrightarrow{e_1}}$  and $C_{\overrightarrow{e_{1,a}}, \overleftarrow{e_1}}$ contain half of a disk of radius 1.

\begin{restatable}{lemma}{lemAngleGadget}
	\label{lem:angle-gadget}
	Given a realization of $P$ and $Q$, assume that $(\overrightarrow{e_{1,a}}, \overrightarrow{e_{1,b}},\overrightarrow{e_{1,c}},\overrightarrow{e_{1,d}})$ lie to the right of $(\overrightarrow{e_1}, \overrightarrow{e_2})$.
	Then, $\overrightarrow{e_1}$ and $\overleftarrow{e_1}$ (resp., $\overrightarrow{e_2}$ and $\overleftarrow{e_2}$) lie exactly on top of each other, and the angle to the right of $(\overrightarrow{e_1}, \overrightarrow{e_2})$ is strictly between $2\cdot\tan^{-1}(1/2)\approx 53.13^\circ$ and $270^\circ$.
\end{restatable}

\begin{proof}
	The fact that $\overrightarrow{e_1}$ and $\overleftarrow{e_1}$ lie exactly on top of each other is a consequence of applying \lemref{lem:realpartiallyfull} to $\overrightarrow{e_1}$ and $\overrightarrow{e_1}'$, and to $\overrightarrow{e_1}'$ and $\overleftarrow{e_1}$. 
	We now focus on the angle constraint.
	Note that by \lemref{lem:realpartiallyfull}, the relative positions of $\overrightarrow{e_1}$ and $\overrightarrow{e_{1,a}}$ (resp., $\overrightarrow{e_2}$ and $\overrightarrow{e_{1,d}}$) is fixed.
	If we fix the positions of $\overrightarrow{e_{1,a}}$ and $\overrightarrow{e_{1,d}}$, then the positions of $\overrightarrow{e_{1,b}}$ and $\overrightarrow{e_{1,c}}$ are completely determined: There are two points whose distance is 3 from the endpoints of $\overrightarrow{e_{1,a}}$ and $\overrightarrow{e_{1,d}}$; one of them causes $\overrightarrow{e_{1,b}}$ and $\overrightarrow{e_{1,c}}$ to intersect with $Q$ which cannot happen since their free space cells are empty.
	If the angle is $2\cdot\tan^{-1}(1/2)$ or smaller, the common endpoint of $\overrightarrow{e_{1,c}}$ and $\overrightarrow{e_{1,d}}$ would lie in the closed $\eps$-neighborhood of $\overrightarrow{e_1}$ and $C_{\overrightarrow{e_1}, \overrightarrow{e_{1,c}}}$ would not be empty (\figref{fig:angle-gadget}(c)), a contradiction.
	If the angle is $270^\circ$ or greater, a portion of $\overrightarrow{e_{1,b}}$ would lie in the closed $\eps$-neighborhood of $\overrightarrow{e_1}$ and $C_{\overrightarrow{e_1}, \overrightarrow{e_{1,b}}}$ would not be empty (\figref{fig:angle-gadget}(d)), a contradiction.
	For all values in between there is a placement for $\overrightarrow{e_{1,b}}$ and $\overrightarrow{e_{1,c}}$ away from $\overrightarrow{e_1}$ and $\overrightarrow{e_2}$, making the section of the free space diagram exactly as required.
\end{proof} 

Using Lemma~\ref{lem:angle-gadget} we can simulate a linkage $\mathcal{L}$ subject to the constraints in Theorem~\ref{thm:abel} using curves given by $D_\eps$, obtaining the following theorem. %(Omitted proofs in Appendix~\ref{app:ER-cont2D}.)

\begin{restatable}{lemma}{thmERcont}
	\label{thm:er-hardness}
	It is $\exists\Reals$-complete to decide if a given free space diagram  is realizable in $\Reals^2$.
\end{restatable}

\begin{proof}
	The described reduction produces a \fsd $D_\eps$ with of size $O(|E(G)|^2)$: $Q$ has length $2|E(G)|$ and each edge in $|E(G)|$ generates up to $10$ segments in $P$, depending on whether the endpoints are rigid or not.
	The runtime is linear in the size of $D_\eps$: each row corresponding to an edge of $P$ has precisely two partially full cells. All other cells are empty.
	
	Given a positive instance of linkage realization, \thmref{thm:abel}(4) and Lemma~\ref{lem:angle-gadget}  guarantee that we can find a placement of $P$ and $Q$ realizing $D_\eps$ as described in the reduction.
	The other direction is a little more subtle. 
	$D_\eps$ forces $Q$ to trace $\sigma$ exactly: using Lemma~\ref{lem:realpartiallyfull} with transitivity constraints the two edges of $Q$ corresponding to an edge in $E(G)$ to lie exactly on top of each other, while the angle gadgets force the circular order around each vertex.
	By Lemmas~\ref{lem:angle-gadget} and \ref{lem:rigid}, $Q$ traces a noncrossing configuration of $\mathcal{L}$ exactly.
	If there is a valid placement of $P$ and $Q$ one can find a noncrossing configuration of $\mathcal{L}$ obtained by the image of $Q$.
	%By Lemmas~\ref{lem:angle-gadget} and \ref{lem:rigid}, $Q$ traces a noncrossing configuration of $\mathcal{L}$ exactly.
	If such a configuration does not satisfy \thmref{thm:abel}(4), that would contradict \thmref{thm:abel}.
	Thus the promise in \thmref{thm:abel}(4) must also be fulfilled by the \frechet realization instance and the angles in each angle gadget would indeed be between $60^\circ$ and $240^\circ$.
\end{proof}

\subsection{Higher Dimensions}
\label{sec:highDimContinuous}

In order to show that free space realization is $\exists\Reals$-hard in higher dimensions, we show that the realizability of linkages and, thus, distance geometry are also $\exists\Reals$-hard. 
Here, the linkage realization is not required to be injective since we are in $\mathbb{R}^{>2}$, but the reduction will force crossings to only happen between predictable pairs of edges. 
This will be important in our reduction to free space realization since crossings between the curves appear in the free space diagram.
We remark that, although all the ingredients of this proof were already known, the claim does not appear in the literature to the best of the authors' knowledge.

\begin{theorem}
	\label{thm:3D} Linkage Realization and Distance Geometry are $\exists\Reals$-hard in  $\Reals^{\ge 2}$.
\end{theorem}
\begin{proof}
	Recall that linkage realization with no pinned vertices is equivalent to distance geometry. 
	The main ingredient of this proof is the \emph{dimension gadget} shown in Figure~\ref{fig:dim-gadget} that appears in~\cite{saxe1979embeddability}.
	%We refer to it as \emph{dimension gadget}.
	The gadget is isomorphic to $K_4$ which is \emph{globally rigid} in $\mathbb{R}^2$~\cite{connelly2005generic}, meaning that there is a unique embedding of the gadget in $\mathbb{R}^2$, and every realization of the gadget in $\mathbb{R}^{>2}$ is congruent with the planar realization.
	Given a linkage $\mathcal{L}$ satisfying the constraints of Theorem~\ref{thm:abel}, replace every edge of $G$ by a copy of the dimension gadget. Note that every vertex $v$ is now represented by two vertices $v_1$ and $v_2$. 
	We call the resulting linkage $\mathcal{L}'$.
	The gadgets force all vertices $v_1$ for all $v\in V(G)$ to be in the same $(k-1)$-hyperplane, perpendicular to the edges $v_1 v_2$. 
	Thus, $\mathcal{L}'$ is realizable in $\mathbb{R}^d$ iff $\mathcal{L}$ is realizable in $\mathbb{R}^{(d-1)}$.
\end{proof}

\begin{figure}
	\centering
	\includegraphics[scale=0.39]{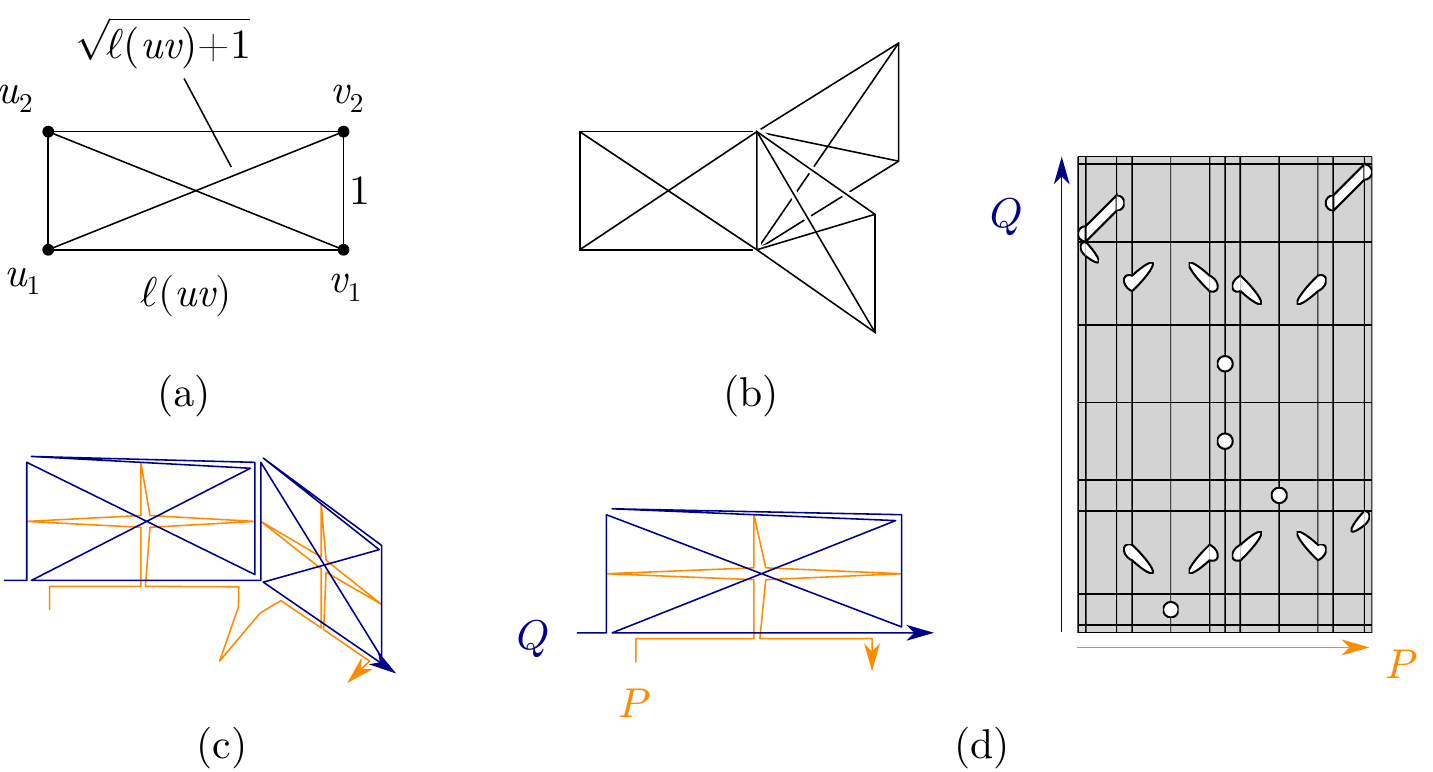}
	\caption{(a) The dimension gadget corresponding to an edge $uv$ with length $\ell(uv)$. (b) Three gadgets in $\mathbb{R}^3$ corresponding to a degree-3 vertex. The gadget forces all vertices to be in one of two planes.
		(c) The realization of an angle gadget after applying the dimension gadget.
		(d) The free space and its realization (perturbed for clarity). %\majid{are fonts in the same scale? some of them seem slightly bigger than others} 
		%\maike{indeed labels could be a bit smaller} 
		%\hugo{I adjusted a bit. Does it look fine now?}
		%\maike{yes, looks good}
	}
	\label{fig:dim-gadget}
\end{figure}

\begin{restatable}{theorem}{thmERdim}
	\label{thm:er-hardness-dim}
	It is $\exists\Reals$-complete to decide if a given free space diagram  is realizable in $\mathbb{R}^{\ge 2}$.
\end{restatable}
\begin{proof}
	We  adapt the dimension gadget to the \fsd realizability problem.
	We first argue for $\mathbb{R}^3$.
	Figure~\ref{fig:dim-gadget}(d) shows the adapted gadget and its free space.
	We use the same reduction as in Lemma~\ref{thm:er-hardness}, but replacing the edges of $Q$ and the edges of $P$ that overlap edges of $Q$ with the modified dimension gadget as follows. 
	Each edge of $Q$ is replaced by a path of length $7$, and each edge of $P$ that lies in the interior of an edge of $Q$ is replaced by a path of length $10$. The free space diagram between the two paths force the path of $Q$ to be embedded as the dimension gadget. 
	Two-dimension gadgets are neighbors if their corresponding edges share an endpoint.
	The cells between these paths and other non-neighbor dimension gadgets are empty. 
	The two length-$\eps$ edges of $P$ in the angle gadget ($\overrightarrow{e_{1,a}}$ and $\overrightarrow{e_{1,d}}$) induce partially full cells in the free space of dimension gadget (first and last collumns in the free space diagram of Figure~\ref{fig:dim-gadget}(d)), forcing these edges to be perpendicular to the plane of the dimension gadget. (See Figure~\ref{fig:dim-gadget}(c))
	The cells of the  two length-$3\eps$ edges of $P$ in the angle gadget ($\overrightarrow{e_{1,b}}$ and $\overrightarrow{e_{1,c}}$) remain empty and thus they must be realized far from the dimension gadget.
	Note that the $\eps$-neighborhood of $v_1$ and $v_2$ does not intersect $P$, and that the $\eps$-neighborhood of the edges of $P$ in the angle gadget still does not intersect $Q$ leaving the dihedral angle between the planes of the neighbor dimension gadgets to vary as in Lemma~\ref{lem:angle-gadget}.
	Thus the embedding of $P$ and $Q$ corresponds to a realization of $\mathcal{L}'$. 
	
	For dimension $d>3$, we recursively apply the dimension gadget construction in the following way. 
	Note that in the $d-1$-dimensional construction each edge of $Q$ overlaps with at least one edge of $P$ (possibly degenerate to a vertex).
	We recursively replace each edge of $Q$ with the dimension gadget construction as normal.
	We split one edge of $P$ that overlaps with the edge of $Q$ at its midpoint and insert the length-8 path forming a cross that contains the midpoints of the edges in the dimension gadget of $Q$ as in Figure~\ref{fig:dim-gadget}(d).
\end{proof}

% \begin{figure}
	%     \centering
	%     \includegraphics[scale=0.5]{figures/dim-gadget-fs.pdf}
	%     \caption{The free space of the dimension gadget and its realization (perturbed for clarity).\maike{add labels P,Q in the figures}}
	%     \label{fig:dim-gadget-fs}
	% \end{figure}

%\input{1DNPhardness}
\section{NP-Hardness and Algorithmic Results for Continuous Curves in $\Reals^1$}
\label{sec:continuous1D}

We briefly consider the realizability problem for curves in $\Reals^1$.  
In this case, the curves have less space to be placed in and hence the free space diagram has limited ``configurations''. 
Cells are still empty, full or partially full cells, but now free space ellipses degenerate to slabs, and the white space is bounded by  parallel line segments  oriented at  $+$ or $-45^\circ$,
% tilted by $\pm 45^\circ$, 
see~\cite{Buchin, Rote}. %and Figure~\ref{fig:partiallycell}.
Here, we present only sketches. For details, we refer to~\cite{eurocg22}. %see Appendix~\ref{app:NPh1D}.
%Recall that for these, partially full free space cells contain slabs oriented at  $+$ or $-45^\circ$, % depending on whether the segments are oriented in the same or opposite directions, 
%see Figure~\ref{fig:partiallycell}.
We note that for curves in 1D the problem is weakly NP-hard by a reduction from the \textsc{Partition} problem. 
The reduction is similar to the hardness of ruler-folding. 

\begin{restatable}{theorem}{onedhardness}
	\label{thm:1D-continuous-hardness}
	Realizability of continuous curves in $\Reals^1$ is weakly NP-complete.   
\end{restatable}

%\maike{I think the details of this section can move to the appendix and be replaced with a sketch, saying that the reduction is similar to ruler folding; and perhaps then merge this into one section with the fpt algorithm in 1D}

%We call a connected number of cells within a single row or column in the diagram a \emph{free space strip}. 
%It corresponds to a single segment of one curve and a connected number of segments of the second curve.

%Next, we prove NP-hardness in \secref{sec:hard}, and give our FPT-algorithm in \secref{sec:alg}.

%\input{1Dfptalgorithm}

Next, we sketch an FPT-algorithm for continuous curves in $\Reals^1$. 
%In this case, the free space diagram has limited ``configurations''. 
%Cells are still empty, full or partially full cells, but now free space ellipses degenerate to slabs and the white space is bounded by  parallel line segments tilted by $\pm 45^\circ$, see~\cite{Buchin, Rote}. 
%In this case, we can infer relative placements of segments already from one point on the free space boundary and cell boundary. 
%
Inspired by computational origami~\cite{Demaine}, we observe %in Appendix~\ref{app:NPh1D} 
that a given diagram $D_\eps$ is realizable iff it can be folded at the grid lines so that the white space is aligned (overlapping only with other white space) into a single convex component.
In~\cite{eurocg22}, we developed an algorithm to enumerate and check the different foldings, inspired by the algorithm for \emph{simple-foldability in $\Reals^1$}~\cite{origami}.
%This enables us to develop an algorithm to enumerate and check the different foldings, inspired by the algorithm for \emph{simple-foldability in $\Reals^1$}~\cite{origami}.
%by Arkin et~al., 
%which asks whether a 1D crease pattern can be folded through a sequence of simple folds ($\pm 180^\circ$ rotations of a portion of the paper).
%Their linear-time greedy algorithm constructively decides whether a given crease pattern admits a sequence of simple folds.
Our algorithm runs in exponential time $O(mn2^k)$, where $k$ is the 
total number of (vertical and horizontal) grid lines of $D_\eps$ that do not intersect the white space 
%
%number of rows of $D_\eps$ that do not intersect the boundary of the white space, i.e., the number of vertical or horizontal strip ``gaps'' 
(completely gray or completely white). 

\begin{restatable}{theorem}{algfptcontinuous}
	\label{thm:alg-fpt-continuous1d}
	Given an $m\times n$ diagram $D_\eps$, in $O(mn2^k)$ time one can find curves $P$ and $Q$ in $\Reals^1$, if they exist, such that $D_\eps=D_\eps(P,Q)$.
\end{restatable}

%\input{1Dpseudopolyalgorithm}
%\section{Pseudo-Polynomial-Time Algorithm for Continuous Curves in $\Reals^1$\label{sec:continuous1Dpseudopoly}}

We now describe an algorithm whose input is a diagram $D_\eps$ where the dimensions of each cell are integers upper-bounded by $W$, and that outputs a pair of curves $P, Q$ in $\Reals^1$ such that $D_\eps=D_\eps(P,Q)$ if they exist.
%If no such pair of curves exist, 
Otherwise, it returns \texttt{false}. 
%Informally, the algorithm again uses the 
We use the %limited configurations in the free space; 
limited placement options of curves in $\Reals^1$;  
the main technical ingredient is the use of dynamic programming to decide the placement of the portions of $P$ and $Q$ for which we have no explicit information.

We use the following placement graph\footnote{This is a variation of the placement graph used for the same problem for continuous curves in $\Reals^2$~\cite{phdleo}.} $G$: %defined in \secref{sec:continuous2DAlgos}. 
The vertices $V(G)$ are the set of segments in the bipartite graph between segments of $P$ and $Q$ where an edge $(v_i^P,v_j^Q)$ encodes that the cell $C_{i,j}$ is partially full. 
$G$ can be computed in $\mathcal{O}(mn)$ time.
By %\lemref{lem:1Dpartfull}
Lemma~3.1 in~\cite{eurocg22}, if $G$ has a single component, we can either compute $P$ and $Q$, or report that no such curves exist in $\Reals^1$ in $\mathcal{O}(mn)$ time. % by computing pairwise the placement of corresponding segments and checking whether every cell in $D_\eps$ is consistent with the placement.
We now show a key property of $G$ for curves in~$\Reals^1$.
We say that a curve in $\Reals^1$ \emph{spans} a distance $w$ if its image in $\Reals^1$ is an interval of length $w$. 
A component of $G$ is a \emph{singleton} component if its size is one (i.e., a single vertex).

\begin{figure}
	\centering
	\includegraphics[width=.7\textwidth]{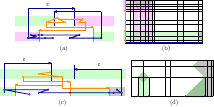}
	\caption{Regions defined by curves $P$ (orange) and $Q$ (blue). Certainty regions are green or pink, subdivision vertices are gray, and uncertainty regions are white (middle) or gray (left/right). 
		%middle uncertainty regions are white, and left/right uncertainty regions are gray. %The right portion of $Q$ that lies to the right of the green rectangle in (c) is in the right uncertainty region of $P$.
	}
	\label{fig:pseudo-poly}
\end{figure}

\begin{restatable}{lemma}{onedtwocomp}
	\label{lem:1D-2comp}
	If $P$ and $Q$ are two curves in $\Reals^1$, then the placement graph $G$ computed from $D_\eps(P,Q)$ has at most two non-singleton components. If either $P$ or $Q$ spans more than $2\eps$, then $G$ has at most one non-singleton component.
\end{restatable}

\begin{proof}
	Let $p_\ell$ and $p_r$ be the leftmost and rightmost points of $P$, respectively.
	We first prove the claim when $P$, without loss of generality, spans more than $2\eps$.
	For contradiction assume there are 2 non-singleton components in $G$. 
	Every point of $Q$ in %the $\eps$-neighborhood of 
	$[p_\ell-\eps, p_r+\eps]$ 
	has exactly distance $\eps$ to some point in $P$ and thus defines the boundary of a free-space component and can be assigned an edge of $G$.
	By continuity, every maximal subcurve of $Q$ in $[p_\ell-\eps, p_r+\eps]$ 
	%the $\eps$-neighborhood of $[p_\ell, p_r]$ 
	corresponds to edges in the same component of $G$. 
	By transitivity, any two overlapping subcurves of $Q$ in %the $\eps$-neighborhood of $[p_\ell, p_r]$ 
	$[p_\ell-\eps, p_r+\eps]$ are also represented in the same component of $G$ as both have at least one point at distance exactly $\eps$ from the same point in $P$. 
	Thus the two components in $G$ correspond to nonoverlapping maximal subcurves of $Q$ in $[p_\ell-\eps, p_r+\eps]$ % the $\eps$-neighborhood of $[p_\ell, p_r]$ 
	and there is no third subcurve of $Q$ that overlaps the first two. %\leo{I tweaked the wording a bit to hopefully make it more clear, please check that I didn't misunderstand anything...}
	Then, $Q$ is disconnected, a contradiction.
	
	Now, consider the case that both $P$ and $Q$ span less than $2\eps$. 
	The points in $\Reals^1$ that are exactly at distance $\eps$ from some point in $P$ form the intervals $[p_\ell-\eps,p_r-\eps]$ and $[p_\ell+\eps,p_r+\eps]$.
	The same argument as above shows that the subcurves of $Q$ in each of these intervals define a single component in $G$.
	Thus, there are at most $2$ non-singleton components. %\leo{I'm not sure I would be able to understand this part without the discussions we had before. Maybe we need to elaborate a bit more?}
\end{proof}

%The proof %(in Appendix \ref{app:ppt1D}) 
%uses that in $\Reals^1$ subcurves connect by spanning the interval between them and therefore the components in the placement graph connect. \\

\noindent\textbf{Algorithm description.}
First, we subdivide the two curves based on the orthogonal projections of the free space's boundary (see Figure~\ref{fig:pseudo-poly}).
We introduce subdivision vertices so that each point in the interior of a segment is either:
\begin{enumerate}
	\item \label{type:far} farther than $\eps$ from any point in the other curve (an edge whose corresponding row or column in $D_\eps$ is completely empty);    
	\item \label{type:close} within $\eps$ distance from every point in the other curve (an edge whose corresponding row or column in $D_\eps$ is completely full); or   
	\item \label{type:boundary} at exactly $\eps$ distance from a point of the other curve (an edge covered by the orthogonal projection of the boundary of the free space).
\end{enumerate}
%\leo{Double-check notation in the end: I think we call edges segments in earlier sections. It makes sense to now talk about edges to distinguish these, but we should make that more explicit.} 
We partition the segments of both curves based on these three types.
Singleton components of $G$ correspond to segments of types (\ref{type:far}) and (\ref{type:close}), and vertices of non-singleton components correspond to segments of type (\ref{type:boundary}).
The first correspond to segments of one curve that are farther than $\eps$ from every point in the other curve.
The presence of type (\ref{type:close}) segments implies that one of the %\leo{both, right?} 
curves spans less than $2\eps$.
Adding subdivision vertices does not asymptotically increase the complexity of the problem as each segment is subdivided at most twice.

For each of the two curves, we partition $\Reals$ into up to $5$ regions used to embed the segments of each of the types based on containment in the $\eps$-neighborhoods of the extreme points of the other curve.
Let $p_\ell$ and $p_r$ be the leftmost and rightmost points of $P$.
%, and recall that $\mathcal{B}_\eps(p_\ell)$ and $\mathcal{B}_\eps(p_r)$ are the balls of radius $\eps$ centered at these points. %\leo{Again, check for consistent notation! I think I denoted disks as $\mathcal{D}_\eps(p)$ for radius $\eps$ and center $p$ at some point...}
We note that we do not have previous knowledge of the points $p_\ell$ and $p_r$ but we later describe how to infer information about these points from $D_\eps$.
The regions serve as an abstraction that allows us to divide the problem into subproblems.
If the balls $\mathcal{B}_\eps(p_\ell)$ and $\mathcal{B}_\eps(p_r)$ intersect, we call the interval $[p_r-\eps,p_\ell+\eps]$ the \emph{middle uncertainty region} (colored white in Figure~\ref{fig:pseudo-poly}(a) and (c)). 
Segments of type (\ref{type:close}) must be embedded in this region.
The intervals of $\Reals^1$ contained in a single disk are called \emph{left}  and \emph{right certainty regions}, respectively $[p_\ell-\eps, p_r-\eps]$ and $[p_\ell+\eps, p_r+\eps]$ (colored green or pink in Figure~\ref{fig:pseudo-poly}(a) and in the bottom curve in (c)).
If $\mathcal{B}_\eps(p_\ell)$ and $\mathcal{B}_\eps(p_r)$ do not intersect, then we call the interval $[p_\ell-\eps, p_r+\eps]$ the \emph{middle certainty region} (colored green in the top curve of Figure~\ref{fig:pseudo-poly}(c)). 
The segments of type (\ref{type:boundary}) must be embedded in these regions.
In both cases, we call the intervals $(-\infty, p_\ell-\eps]$ and $[p_r+\eps, \infty)$ the \emph{left} and \emph{right uncertainty regions} (colored gray in Figure~\ref{fig:pseudo-poly}. Note that the figure only shows a closeup view and the only visible portion of a right uncertainty region is shown for the bottom curve in (c)).
The segments of type (\ref{type:far}) must be embedded in these regions.

\begin{restatable}{lemma}{onedseperate}
	\label{lem:1D-separation}
	Given $D_\eps$, in $O(nm)$ time we can partition the segments of $Q$ into the three types and assign each segment to a region.  
\end{restatable}

\begin{proof}    
	The orthogonal projection of the components can be computed by a traversal of the free space's boundary. 
	Thus, we can compute the subdivision vertices in $O(nm)$ time.
	The types of all segments can be inferred from $D_\eps$ in $O(nm)$ time. %using the definition.
	We can compute $G$ in $O(nm)$ time. If there are two non-singleton components, we arbitrarily fix the orientation of one segment and use %\lemref{lem:1Dpartfull} 
	Lemma~3.1 in~\cite{eurocg22} to decide the relative placement for their respective segments. %in such components. 
	Note that the type of the segments adjacent to a segment in a certainty region determines which of the components is the left and which is the right certainty region: the left certainty region is adjacent to segments of type (\ref{type:close}) on its right boundary. 
\end{proof}

We can use %Lemma~\ref{lem:1Dpartfull} 
Lemma~3.1 in~\cite{eurocg22} to determine the relative embedding of $P$ and $Q$ in certainty regions.
It remains to determine whether the subcurves in uncertainty regions can be embedded.
We use a dynamic program (DP) to solve the problem in each uncertainty region separately.
We further divide the problem into two cases depending on whether we know the relative position of the boundaries of the respective region.
The input of each DP is a maximal subcurve of $P$ (resp., $Q$) in an uncertainty region.
The DP computes the possible placements of the subcurve for a set of boundary constraints.
We later describe how to combine the output of all the DPs into a single solution. 

\smallskip\noindent\textbf{Fixed boundary subproblem.}
%Full details are given in Appendix~\ref{app:fixed-boundary}.
If one of the curves does not have edges of type (\ref{type:close}), by %Lemma~\ref{lem:1Dpartfull} 
Lemma~3.1 in~\cite{eurocg22} we know the size of the uncertainty regions.
We define DP problems for each maximal subcurve in an uncertainty region whose value is \texttt{true} iff it is possible to realize the subcurve in the region.
%or \texttt{false} otherwise. 
%For that, 
We define subproblems based on a suffix of the subcurve and the coordinate of the first point of the suffix (details are in Appendix~\ref{app:fixed-boundary}). 
The recursive definition tries embedding the next edge oriented towards the right or left, thus each subproblem depends on only two  subproblems.
The number of subproblems depends on the number of segments in the subcurve and the size of the uncertainty region. %Thus we get
%the following lemma.

\begin{restatable}{lemma}{fixBoundary}\label{lem:fixBoundary}
	One can compute each fixed boundary subproblem defined by a subcurve $Q'$ with $n'$ segments, each with an integer length of at most $W$, and an integer interval $[0,r]$, in  $O(n'\cdot\min(r, n'W))$ time.    
\end{restatable}

\begin{proof}
	Since $k\in\{1,\ldots,n'\}$ and $s\in\{0,\ldots,r\}$ there are at most $O(n'r)$ subproblems.
	We can also upper-bound $s$ by $n'W$ since this is the maximum length of the image of $Q'$ (which is necessary in the case when $r=\infty$).
	Each subproblem can be computed in $O(1)$ time.
	Thus the total runtime is $O(n'\cdot\min(r, n'W))$.
\end{proof}

\smallskip\noindent\textbf{Variable boundary subproblem.}
%Full details are in Appendix~\ref{app:var-boundary}.
When both  curves have segments of type (\ref{type:close}), the size of the middle uncertainty region of one curve depends on the size of the middle uncertainty region of the other.
We similarly define a DP problem for each maximal subcurve in an uncertainty region. 
However, each subproblem is also defined by a suffix of the subcurve, the coordinate of the first point, and, additionally, the size of the uncertainty region.
In this case, the size of the uncertainty region is upper-bounded by $2\eps$. 
%Then, we have the following lemma.
\begin{restatable}{lemma}{varBoundary}\label{lem:varBoundary}
	For variable boundary subproblems, in $O(\max(n,m)\cdot\eps^2)$ time, one can compute all DP tables and, if there exist $P$ and $Q$ in $\Reals^1$ that realize $D_\eps$,  find $r_P$ and $r_Q$ that are compatible with a solution to all subproblems, where $r_P$ and $r_Q$ denote the sizes of the middle uncertainty regions of $P$ and $Q$ respectively.    
\end{restatable}

\begin{proof}    
	There are $O(m+n)$ DP tables since this is the upper bound on the  number of maximal subcurves in the middle uncertainty regions. 
	Each table has $O(n'\cdot\eps^2)$ subproblems, each can be computed in $O(1)$ time, where $n'$ is the size of the maximal subcurve.
	Thus, it takes $O(\max(n,m)\cdot\eps^2)$ to compute all DP tables.
	For each table, we can keep track in a separate data structure what values of $\alpha$ have an entry $R(i, .,\alpha)=\texttt{true}$.
	Then, given values for $r_P$ and $r_Q$, we can check whether there exist a compatible solution in each table in $O(1)$ time per table.    
	Thus, we can try all possible $r_P, r_Q\in\{1,\ldots,2\eps\}$ searching for values compatible with a solution for each DP problem. 
	Then, searching for a set of compatible solutions takes $O(\max(n,m)\cdot\eps^2)$ time.
\end{proof}

% \todo{reviewer comment: It seems that for Lemma 17 and 18 we can just brute force over all folding patterns for the uncertain parts as there can be at most a polynomial number of positions where we can end up. Why not just explicitly state this?}

%Combining Lemmas~\ref{lem:fixBoundary} and \ref{lem:varBoundary}, we have the following theorem.

\begin{theorem}
	\label{thm:pseudo-poly}
	Given an $m\times n$ free space diagram $D_\eps$, where $n\ge m$, every cell has integer dimensions of at most $W$, and $\eps$ is an integer, we can produce two curves in $\Reals^1$ that realize $D_\eps$ or answer \texttt{false} if no such curves exist in time $O(\max(n\eps^2, n^2W))$.
\end{theorem}
\section{$\exists\Reals$-Completeness for Discrete Curves in $\Reals^2$\label{sec:discrete2DHardness}}

We now turn to the discrete \fd, and prove that realizability by curves in $\Reals^{\geq 2}$ for a given free space matrix is also $\exists\Reals$-complete.
We reduce from \textsc{$d$-Stretchability}, which asks whether there exists an arrangement of hyperplanes in $\Reals^{\ge 2}$ that realizes a combinatorial description. The formal description follows in the next paragraph.
We use the machinery developed by Kang and M{\"u}ller~\cite{kang2012sphere} to show that recognizing a $d$-sphere graph (generalization of unit disk graphs in $\Reals^d$) is $\exists\Reals$-hard for $d\ge 2$.
(Although only NP-hardness is claimed in~\cite{kang2012sphere}, their proof also extends to $\exists\Reals$-hardness as noted in their conclusion.)
%\hugo{This will probably also move to intro.}
Recall that we explain in Section~\ref{sec:intro} that, though they are similar, our problem differs from $d$-sphericity.
%Though our problem is similar to $d$-sphericity, we note that the problems are indeed different, 
%as explained in Section~\ref{sec:intro}. 
%
%A unit disk graph requires that two  vertices are adjacent iff they  are placed within a distance $\le 1$.
%Our problem can be seen as given by a bipartite graph where every pair of adjacent vertices is placed within a distance $\le 1$. However, the placement of vertices of the same partite set is unconstrained, i.e., they are not adjacent (by bipartiteness) and \textit{can} be placed within a distance $\le 1$.
%In other words, if an instance of our problem is also a unit disk graph, then it admits a positive answer, but if there are curves $P$ and $Q$ that realize a free space matrix, our answer does not necessarily correspond to a unit disk graph. 

An instance of \textsc{$d$-Stretchability} is given by a set $S\subseteq\{-,+\}^n$ of size $1+\binom{n+1}{2}$.
%A line arrangement is simple if no two lines are parallel, no three lines intersect at the same point, and no lines are vertical.
An  arrangement of $n$ hyperplanes divides $\Reals^d$ into $1+\binom{n+1}{2}$ cells.
Each vector in $S$ corresponds to a cell in a potential arrangement.
We denote by $\textbf{v}_j\in S$ the $j$th vector in $S$ and by $\textbf{v}_j[i]$ its $i$th coordinate.
Then $\textbf{v}_j[i]=-$ (resp., $\textbf{v}_j[i]=+$) if the corresponding cell is below (resp., above) the $i$th hyperplane.
Note that $(-,\ldots,-)$ and $(+,\ldots,+)$ must be in $S$, and so we assume they are respectively $\textbf{v}_1$ and $\textbf{v}_2$.
The problem asks whether $S$ is the combinatorial description of an arrangement of $n$ hyperplanes.

\noindent\textbf{Reduction.}
Given an instance $S$ of \textsc{$d$-Stretchability} of $n$ hyperplanes we construct a $2n\times |S|$ free space matrix $M_{\eps}$ as follows.
We partition $P$ (whose vertices correspond to rows of $M_{\eps}$) into two subcurves $P^+=(a_1,\ldots,a_n)$ and $P^-=(b_1,\ldots,b_n)$.
Informally, $a_i$ (resp., $b_i$) will be a point in the upper (resp., lower) halfspace of a hyperplane $\ell_i$ in the arrangement.
Each column $j$ of $M_{\eps}$ (i.e., vertex of $Q$) represents a vector in $\textbf{v}_j\in S$, that is,  $M_{\eps}[i][j]=0$ and $M_{\eps}[n+i][j]=1$ (resp., $M_{\eps}[i][j]=1$ and $M_{\eps}[n+i][j]=0$) if $\textbf{v}_j[i]=-$ (resp., $\textbf{v}_j[i]=+$).

\begin{restatable}{theorem}{ERdiscrete}
	\label{thm:ER-discrete}
	Given a free space matrix $M_{\eps}$, it is $\exists\Reals$-complete to decide whether there exists a pair of curves $P$ and $Q$ in $\Reals^{\ge 2}$ that realizes $M_{\eps}$.
\end{restatable}

\begin{figure}[ht]
	\centering
	\includegraphics[width=\textwidth]{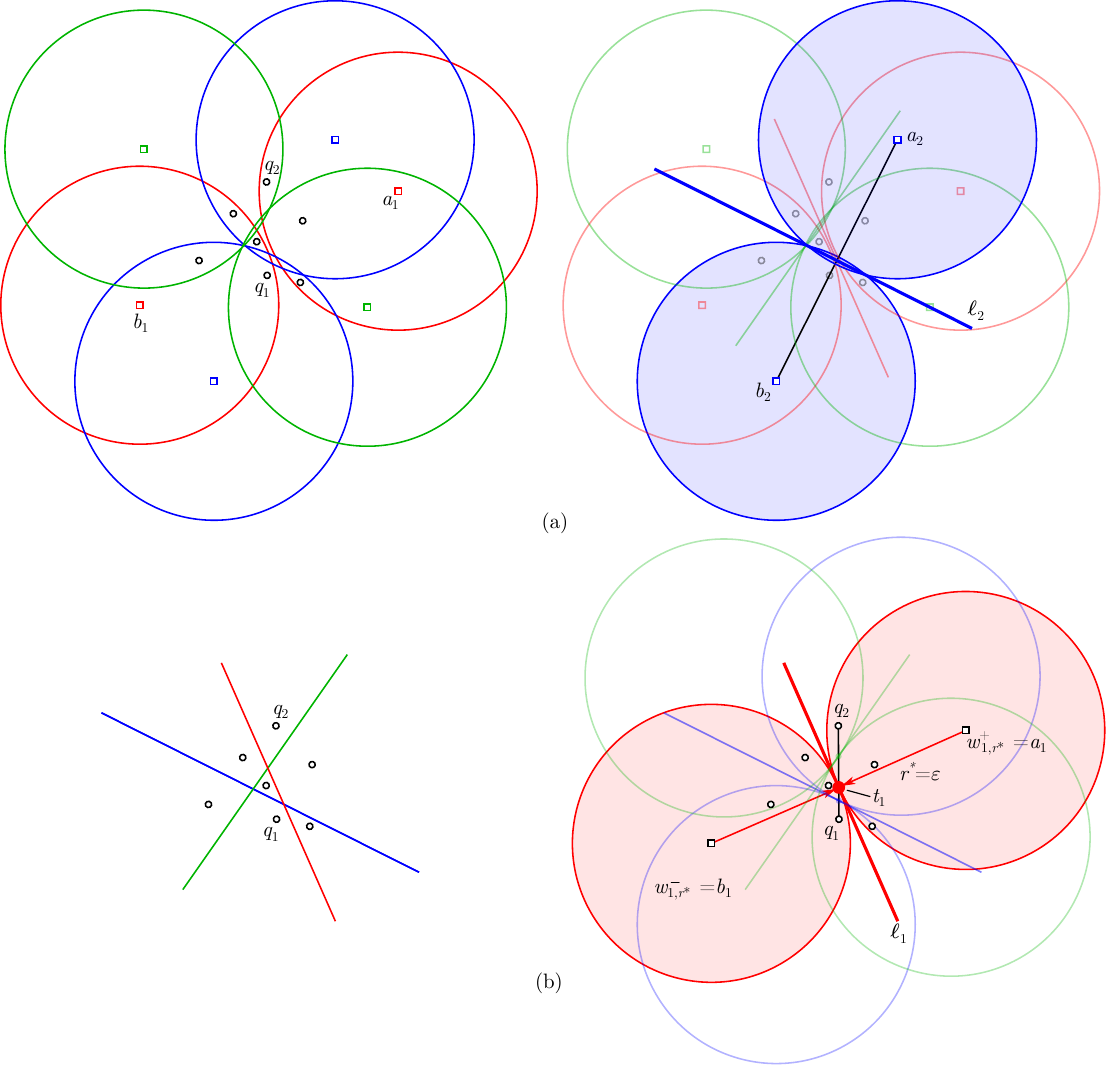}
	\caption{Reduction from $S=\{(-,-,-),(+,+,+), (-,-,+), (-,+,+), (-,+,-), (+,+-), $ $(+,-,-)\}$. (a) Transforming a solution to $M_\eps$ into a line arrangement that realizes $S$.
		(b) Transforming a solution to $S$ into curves $P$ and $Q$ that realize $M_\eps$. Here squares represent points of $P$ and circles represent points of $Q$.}
	\label{fig:er-discrete}
\end{figure}

\begin{proof}
	Containment in $\exists\Reals$ can be proven by a straightforward reduction to $\exists\Reals$ similar to the proof of Lemma~\ref{lem:ER-containment-continuous}.
	We now focus on the reduction defined above.
	It is clear that it runs in polynomial time.
	Assume that there exists a pair of curves $P$ and $Q$ that realizes $M_{\eps}$.
	Refer to Figure~\ref{fig:er-discrete}(a).
	We use the labels of point of $P$ defined in the reduction and assume $Q=(q_1,\ldots, q_{|S|})$.
	Recall that, informally, points $q_1$ and $q_2$ represent vectors $\textbf{v}_1=(-,\ldots,-)$ and $\textbf{v}_2=(+,\ldots,+)$, respectively.
	Rotate the solution in order to make the vector $\overrightarrow{q_1q_2}$ vertical and pointing upwards.
	We build a hyperplane arrangement as follows.
	For each $i\in [n]$, create a hyperplane $\ell_i$ bisecting the segment $a_i b_i$.
	Now, we argue that $q_j$, $j\in\{1,\ldots,|S|\}$ is in a cell in the produced arrangement with description $\textbf{v}_j$.
	Let $C_1$ and $C_2$ be the cells in the arrangements of circles of radius $\eps$ containing $q_1$ and $q_2$, respectively.
	By definition, if $\textbf{v}_j[i]=+$, then $q_j$ must be within $\eps$ distance from $a_i$ and farther than $\eps$ from $b_i$, that is $q_j\in \mathcal{B}_\eps(a_i)\setminus\mathcal{B}_\eps(b_i)$.
	Thus, $C_1= (\bigcap_{i=1}^n \mathcal{B}_\eps(b_i)\setminus \bigcup_{i=1}^n \mathcal{B}_\eps(a_i))$ and $C_2= (\bigcap_{i=1}^n \mathcal{B}_\eps(a_i)\setminus \bigcup_{i=1}^n \mathcal{B}_\eps(b_i))$.
	Note that every hyperplane $\ell_i$ must separate $C_1$ and $C_2$ by definition.
	Thus every $\ell_i$ intersects the line segment $\overline{q_1q_2}$.
	We focus on a specific hyperplane $\ell_i$.
	Without loss of generality assume $\textbf{v}_j[i]=+$.
	Then, $\mathcal{B}_\eps(a_i)\setminus\mathcal{B}_\eps(b_i)$ is above $\ell_i$ and so is $q_j$.
	Therefore, the produced hyperplane arrangement realizes $S$.
	
	Now assume that there exists a hyperplane arrangement realizing $S$.
	Refer to Figure~\ref{fig:er-discrete}(b).
	For each cell in the arrangement described by $\textbf{v}_j$, choose a point $q_j$ in the interior of the cell.
	As before, every hyperplane intersects the line segment $\overline{q_1q_2}$, since $q_1$ is below all the hyperplanes and $q_2$ is above.
	Let $t_i$ be the intersection of $\ell_i$ and $\overline{q_1q_2}$.
	Define the balls $\mathcal{B}_r(w_{i,r}^+)$ and $\mathcal{B}_r(w_{i,r}^-)$ respectively above and below $\ell_i$, tangent to $\ell_i$ at $t_i$.
	Note that $\mathcal{B}_r(w_{i,r}^+)$ (resp., $\mathcal{B}_r(w_{i,r}^-)$) equals the upper (resp., lower) halfspace of $\ell_i$ when $r\rightarrow \infty$.
	Thus, for sufficiently large $r$, $\mathcal{B}_r(w_{i,r}^+)$ contains all points $q_j$ above $\ell_i$ and $\mathcal{B}_r(w_{i,r}^-)$ contains all points $q_j$ below $\ell_i$.
	Let $r^*$ be a sufficiently large $r$ such that the previous statement is true for all $i\in [n]$.
	Scale the entire construction to make $r^*=\eps$.
	Then, we can construct $P$ by making $a_i=w_{i,\eps}^+$ and $b_i=w_{i,\eps}^-$.
	Now, each $q_j$ is contained in the appropriate cell of the arrangement of circles of radius $\eps$ centered at points of $P$.
	Thus, the constructed $P$ and $Q$ realize $M_\eps$.
\end{proof}

%\begin{proof}[Proof sketch]    
%The construction forces the intersection $\mathcal{B}_\epsilon(a_i)\cap \mathcal{B}_\epsilon(b_i)$ to be empty of points of $Q$ while the union $\mathcal{B}_\epsilon(a_i)\cup \mathcal{B}_\epsilon(b_i)$ contains all points of $Q$.
%Intuitively, $\mathcal{B}_\epsilon(a_i)$ (resp., $\mathcal{B}_\epsilon(b_i)$) represents the lower (resp., upper) halfspaces of the hyperplane $\ell_i$ in the arrangement containing the appropriate points of $Q$. 
%Thus, the position of $a_i$ and $b_i$ encode the position of $\ell_i$. 
%The other direction, i.e., finding $P$ from an arrangement of hyperplanes that divide points in $Q$ in the desired way, is a bit more involved.
%We choose a point in each hyperplane and define tangent balls $\mathcal{B}_\epsilon(a_i)$ and $\mathcal{B}_\epsilon(b_i)$. We scale the construction so that all points of $Q$ lie in $\mathcal{B}_\epsilon(a_i)\cup \mathcal{B}_\epsilon(b_i)$.
%A more detailed proof is given in Appendix~\ref{app:discrete}.
%\end{proof}

%\todo{reviewer comment: I think this proof sketch is too brief. There is a scaling argument that has to be made that seems interesting to me.}
%\maike{indeed the sketch is very short; if we have space, we could expand a bit}
%\leo{I simply moved the full proof here, as we have the space within the first 16 pages and safe space in total.}

%\input{discreteFrechet} %% 
\section{A Polynomial Time Algorithm for Discrete Curves in $\Reals^1$\label{sec:discrete1Dpoly}}
We now turn to realizability for discrete curves in $\Reals^1$ and show that this can be decided in polynomial time.  %Given a free space matrix $M_{\eps}$ consisting of entries filled with either 1 or zero, and a threshold $\eps = 1$, the aim is to compute two curves $P$ and $Q$ with numbers of vertices $n$ and $m$, resp., so that the proximity between their vertices fulfills the corresponding entries of the matrix $M$. An entry $m_{i,j}$ in $M$ is filled with $1$, i.e.,  true $(\true)$ if $\|p_i-q_j\|\leq 1$ and $0$, i.e.,  false $(\false)$, otherwise. 
We lend our main idea from the {\em unit-interval graph recognition} (\textsc{UIGR}) in~\cite{CORNEIL}: given an abstract graph $G$ whose nodes are intervals and two intervals intersect iff there is an edge between the two corresponding nodes in $G$, the goal is to find a placement of intervals in the real line such that the intersections induced by them fulfill $G$. See \figref{fig:uigr} for an example. 
%of \textsc{UIGR}. 
In free space realizability, we derive a unit interval graph $G$ from the free space matrix $M_\eps$. We then adapt the idea in~\cite{CORNEIL} to handle the realizability in our case. %In the \textsc{UIGR}, the aim is to find an ordering of intervals in real line that fulfills $G$ 

\begin{figure}[htpb]
	\centering
	\includegraphics[width=.35\textheight]{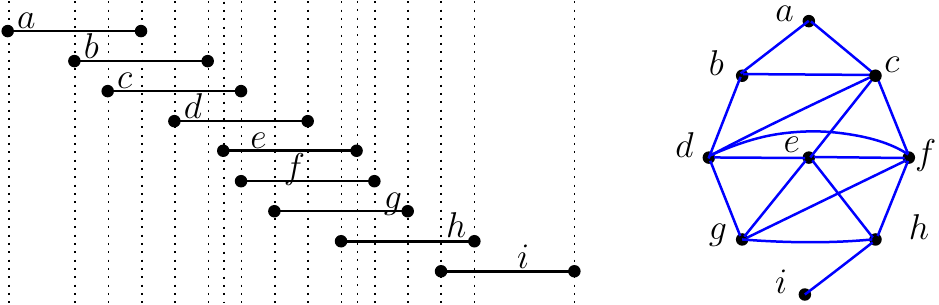}
	\caption{An abstract graph of intervals and its recognition in $\Reals^1$}
	\label{fig:uigr}
\end{figure}

%For the sake of comprehensive presentation, we first describe their algorithm and then show how we modify it.  
First, we borrow some notations from~\cite{CORNEIL}. For a vertex $v \in G$ the neighboring vertices of $v$ is denoted by $N(v)$. We also define $N[v] = N(v) \cup \{v\}$. Two vertices $u$ and $v$ are {\em indistinguishable} if $N[u] = N[v]$.
In order to recognize a unit-interval graph, Corneil et al.~\cite{CORNEIL} propose a linear-time algorithm: (i) find the left anchor in $G$ (the left-most interval in the recognition), (2) perform a BFS search starting at the left anchor to get a {\em partial order} of the intervals, meaning that some groups of intervals are ordered properly, but still intervals belonging to each group need to be re-ordered, and (3) refine the partial order, i.e., the intervals within in each group, to get the {\em global order}. %Since any solution to free space matrix realizability is a solution to the \textsc{UIGR} we handle (3) with further consideration compared to~\cite{CORNEIL}.
If the global order exists, return `YES', and `NO', otherwise. 
To handle (1), they perform a BFS from an arbitrary node in $G$. Find a vertex $z$ at the last level $L_t$ of the BFS search such that $\deg(z) = \min\{\deg(w): w \in L_t\}$. In (2), they locally order the vertices in $G$ based upon the level in which they are encountered along the BFS search from $z$. Finally in (3), in each level $L_k$ %of the BFS search 
obtained in (2)
they sort each vertex $v \in L_k$ in an increasing order of $D(v) = |\nex(N(v))| - |\pre(N(v))|$. Here, $\nex(v)$ is the set of adjacent vertices to $v$ %at level $k+1$ of the BFS 
in $L_{k+1}$
and $\pre(v)$ is the set of adjacent vertices to $v$ %at level $k-1$ of the BFS. 
in $L_{k-1}$.
Next, for each vertex $v$ in $G$, they compute $\alpha(v)= \min\{\mbox{order}(u):~u \in N[v]\}$ and $\omega(v)= \max\{\mbox{order}(u):~u \in N[v]\}$, where $\mbox{order}(u)$ is the order in which $u$ is encountered on the line. In the end, if there exists a $v$ for which $\alpha(v) \neq \omega(v)$, `NO' is returned, and `YES' is returned, otherwise. We modify Step (3) to handle our case.  

%Below we present our modified algorithm that uses the notions presented in the algorithmic description above. 
The problem of realizing $M_\eps$ is more restrictive than UIGR as follows.
%For simplicity, l
Let $\eps=1/2$.
Similar to the $\Reals^2$ case, we can see the realization of $Q=(q_1,\ldots,q_n)$ as an arrangement of intervals $\mathcal{B}_\eps(q_j)$ partitioning $\Reals^1$.
Each row $i$ of $M_\eps$ describes a ``cell'' in the arrangement (which is an interval in $\Reals^1$) where we place a point $p_i$ of $P$. 
Thus, $M_\eps$ requires a unit interval realization of $P$ with not only prescribed adjacencies in $G$ but prescribed cells.
Our goal is to take the partial order that we get from (2), and refine it to be closer to a global order using information from $M_\eps$. 
%In that process we will order most indistinguishable vertices; a few will remain indistinguishable, we call those {\em ``equivalence classes''}. 
In the following, we refer to vertices of $G$ and the interval they refer to interchangeably, and we call the maximal set of vertices that are pairwise indistinguishable an \emph{equivalence class}. See Appendix~\ref{app:discretePoly} for a full algorithmic description. 
%In the algorithm below, steps 2-4: we still have a couple of vertices with same neighborhood set ("indistinguishable"). And our goal is to disambiguate them (meaning, give them an order -- could also say we want to "order" them) so that we arrive at something that's almost a global order (a few might still stay indistinguishable, we call those equivalence classes).

%\begin{comment}
\noindent$\textsc{DiscreteRealizabilityAlgo}(M_\eps):$
%\carola{Goal: We want to take the partial order that we get from the initial BFS, and refine it to be closer to a global order. In that process we will order most indistinguishable vertices; a few will remain indistinguishable, we call those "equivalence classes". }
%\carola{Terminology: Indistinguishable: BFS steps 2-4: still have a couple of vertices with same neighborhood set ("indistinguishable"). And our goal is to disambiguate them (meaning, give them an order -- could also say we want to "order" them) so that we arrive at something that's almost a global order (a few might still stay indistinguishable, we call those equivalence classes). }
{\bf (1)} Construct the unit-interval graph $G$ from $M_{\eps}$. %whose rows are vertices in $P$ and columns are vertices in $Q$: The vertex set of $G$ represents the columns of $M_{\eps}$, i.e., the interval of length $2\eps$ centered at a point $q\in Q$. For every row ($p \in P$), the edge set of $G$ is defined by including a clique between the columns (vertices in $Q$) that are filled with $1$ in that row. We store the edges of $G$ in an adjacency matrix. In the following we assume $G$ connected, or else, we treat each component separately.
%\carola{ Then $|E|=O(m^2)$ storage but we need $O(k^2n)$ time to add all edges. The number of edges $|E|=min(m^2, k^2n$). Or, we could live in **adjacency lists** world, add in duplicate edges, and then sort them to get rid of duplicates -> Runtime (edge generation + sorting) $k^2n + k^2n\log(k^2n)$ or with **counting sort**: $k^2n + k^2n + m$ and since $m<n$ -> $O(k^2n)$.}
{\bf (2)}
Choose a left anchor $v_0$ by running a BFS search on $G$.
%Since all candidates for the left anchor are indistingushable, we can compute all such candidates by performing a BFS search on $G$. The ultimate left anchor is obtained by connecting a new vertex $v_0$ to the candidates. 
%by performing a BFS search on $G$ while trying all ties. 
%\carola{BFS takes time $O(|V|+|E|)=O(m + m^2)$ if in adjacency matrix form, or if we convert to adjacency lists then $O(m+min(m^2, k^2n)$}
{\bf (3)}  Perform a BFS on $G$ starting at $v_0$ to obtain a partial order of the intervals.  
{\bf (4)} Refine the partial order %to get the global order 
by sorting the intervals at each level of the BFS under the criterion of $D(v)=|\nex(N(v))|- |\pre(N(v))|$. 
{\bf (5)} 
Refine the partial order $D$ to an order $D'$: 
For each row of $M_\eps$, place intervals of entry 1 in the equivalence class $C$ towards those intervals that don't belong to $C$ whose entries are 1 and orders are different than intervals in $C$.
{\bf (6)} Extend the partial order defined by the BFS levels and $D'$ to a global order breaking ties arbitrarily.
%Obtain the arrangement of unit intervals based on the global order and $G$. This can be done as in Theorem 3.2 in~\cite{CORNEIL}.
%obtained order with respect to $D'(v)$ for all intervals $v \in G$ into a global order arbitrarily; those intervals $v$ that are indistinguishable (with identical $D(v)$) can be replaced. 
{\bf (7)} Verify whether the produced arrangement is compatible with $M_\eps$ or not. 
%For each row ($p \in P$), return `YES' if the existence of the interval containing the correct labeling from the unit intervals is identified, and return `NO', otherwise. 
%\end{comment}
\begin{restatable}{lemma}{Discrete} \label{lem:uig-fsm}
	The algorithm returns `YES' if and only if $M_\eps$ is realizable.
\end{restatable}

\begin{proof}
	By Step (7), it is clear that when the algorithm returns `YES' the instance $M_\eps$ is realizable. 
	If the algorithm returns `NO', we show that there are no $P$ and $Q$ realizing $M_\eps$. 
	Note that the constraints in the ordering obtained from Steps (1--4) are the same as for UIGR.
	Thus, we must show that if $G$ is not a unit interval graph, then $M_\eps$ is not realizable.
	We show the constrapositive: if $M_\eps$ is realizable, then $G$ is a unit interval graph.
	As discussed before, the realization of $Q$ implies the realization of a unit interval graph $G^*$ whose intervals are $\mathcal{B}_\eps(q_j)$, for $\eps=1/2$.
	It is clear that $G^*$ is a supergraph of $G$ since the required cells in the interval arrangement exist containing points of $P$.
	Let $q_iq_j$ be an edge that exists in $G^*$ and not in $G$ such that $q_i$ is to the left of $q_j$ and with shortest interval $\mathcal{B}_\eps(q_i)\cap \mathcal{B}_\eps(q_j)$.
	Since $q_iq_j$ is not in $G$, $\mathcal{B}_\eps(q_i)\cap \mathcal{B}_\eps(q_j)$ contain no points of $P$ and all the cells in this interval are not necessary.
	By the assumption that $\mathcal{B}_\eps(q_i)\cap \mathcal{B}_\eps(q_j)$ is shortest, there is no point $q_k$ of $Q$ to the right of $q_j$ whose interval $\mathcal{B}_\eps(q_k)$ intersects $\mathcal{B}_\eps(q_i)\cap \mathcal{B}_\eps(q_j)$.
	Let $S$ be the set of points including $q_j$ and all points $q_k$ of $Q$ to the right of $q_j$ whose interval $\mathcal{B}_\eps(q_k)$ does not intersect $\mathcal{B}_\eps(q_j)$. 
	Move all points in $S$ until the intersection $\mathcal{B}_\eps(q_i)\cap \mathcal{B}_\eps(q_j)$ disappears.
	By construction, ass cells previously in $\mathcal{B}_\eps(q_i)\cap \mathcal{B}_\eps(q_j)$ disappear.
	We show that no other cell does, and thus the modified $Q$ is still a solution for $M_\eps$.
	A cell to the left of $q_j-\eps$ is not affected since we only move points to the left of $q_j$.
	A cell to the right of $q_j+\eps$ would disappear if $\mathcal{B}_\eps(q_j)$ starts to intersect with an interval that it didn't before.
	This does not happen since $S$ contains all such intervals.
	A cell in $\mathcal{B}_\eps(q_j)\setminus \mathcal{B}_\eps(q_i)$ would disappear if an interval defined by $q_k\in S\setminus \{q_j\}$ stopped intersecting another interval. 
	Recall that there are no intervals whose right endpoint are between $q_j+\eps$ and $q_j+\eps+|\mathcal{B}_\eps(q_i)\cap \mathcal{B}_\eps(q_j)|$ by the ``shortest'' assumption.
	Then, such cells do not exist.
	Thus, we produced a solution to $M_\eps$ whose interval intersection graph has fewer edges than $G^*$.
	Applying this argument successively we conclude that there exist a solution whose intersection graph is $G$.
	
	We now show that Step (5) is necessary.
	Suppose there are four intervals $\{a,b,c,d\}$ in row $r$, $I_r = \{a, c\}$, $C = \{b,c,d\}$, and $C' = \{c\}$. Also by assumption %$D(a) < D(b) = D(c) = D(d)$.
	$a <_D b =_D c =_D d$.
	For the sake of contradiction, suppose that there exists a positive solution that does not place the interval $c$ before $C\backslash C' = \{b, d\}$ and after $I_r\backslash C = \{a\}$. Placing $c$ before $a$ implies that %$D(c) \leq D(a)$ 
	$c <_D a$
	which is a contradiction.
	Placing $c$ after $b$ implies that %$c$ and $a$ should not intersect $b$ but each other. 
	the intersection $\mathcal{B}_\eps(a)\cap \mathcal{B}_\eps(c)$ is contained in $\mathcal{B}_\eps(b)$.
	This means that the cell required by $I_r = \{a, c\}$ does not exist, which is again a contradiction.  
	
	Finally, we analyse Steps (6--7).
	The only worry is that there might exist two  extensions of the partial order produced in Step (5) such that one is a positive solution and the other isn't.
	Let $q_i$ and $q_j$ be two incomparable vertices in the partial order. 
	Since Step (5) refines equivalence classes, $q_i$ and $q_j$ are also in the same equivalence class $C$.
	Then, there exist no row containing an interval not in $C$ whose entries relative to $q_i$ and $q_j$ are different.
	If there is a row containing only a proper subset of $C$, the arrangement must contain a cell in which the proper subset of $C$ intersect and that does not intersect any other interval. 
	Since intervals in $C$ intersect the same intervals by definition, including intervals that must be to the left and to the right of intervals in $C$ (note that we add $v_0$ and $v_f$ so that every equivalence class has a predecessor and successor), this cell cannot exist.
	Then Step (7) returns `NO'.
	Else, the columns relative to $q_i$ and $q_j$ are identical and interchangeable.
\end{proof}
%\begin{proof}[Proof sketch]
%By Step (7), it is clear that if the algorithm returns `YES' then $M_\eps$ is realizable. 
%Then we show that all constraints in the ordering obtained from all the steps are necessary.
%\end{proof}
The construction of $G$ takes $O(k^2n)$ time where $k\le m$ is the maximum number of entries filled with 1 over all rows in $M_\eps$, since $G$ might contain $n$ cliques of size $k$. 
The other steps can be implemented in linear time. The full algorithm description is given in  Appendix~\ref{app:discretePoly}. 

\begin{restatable}{theorem}{Discreteruntime}\label{thm:uig-runtime}
	Given an $n\times m$ free space matrix $M_\eps$, we can decide whether there exist curves $P$ and $Q$ in $\Reals^1$ that realize $M_\eps$ in  $O(nm+k^2n)$ time, where $m\leq n$. %and $k\le m$ is the maximum number of entries filled with 1 over all rows in $M_\eps$.
\end{restatable}
\begin{proof}
	%Let $T(n,m,k)$ be the runtime of the algorithm. %W.l.o.g. assume that $m\leq n$. 
	Correctness is given by Lemma~\ref{lem:uig-fsm}.
	First note that the size of $V$ derived from $M_\eps$ is $|V| = m$, and $|E|= O(m^2)$ due to the size of the adjacency matrix we use in Step (1), however, we need $O(k^2n)$ time to add all edges in $G$ due to the size of the clique induced by the intervals with entry 1 in each row. The size of each clique is at most $k^2$. Thus 
	$|E| = \min(m^2,k^2n)$. % and $T(n,m,k) = O(k^2n)$ and 
	%Step (1) takes $O(|V| + |E|) = O\big(m+ \min(m^2,k^2n)\big)$ time.
	Steps (2--3) takes $O(|V| + |E|) = O\big(m+ \min(m^2,k^2n)\big)$ for performing the BFS on $G$. %the size of equivalence class is $O(m)$. candidate nodes for identifying each of $O(m)$ many indistinguishable vertices.
	%Step (3) simply takes $O(|V| + |E|) = O\big(m+ \min(m^2,k^2n)\big)$. 
	Step (4) takes $O(|V|) = O(m)$ time using a counting sort per level of the BFS. 
	Step (5) takes $O(mn)$ time by processing each row of $M_\eps$ and 
	partitioning the relevant intervals into their equivalence classes in $O(m)$ time.
	%checking the order of each $O(m)$ vertices that are in equivalence class. 
	%Hence $T(n,m,k) = O(mn+ \min(m^2,k^2n))$.
	Step (6) takes $O(m)$ time. 
	In step (7), the real line can be partitioned into $2m = O(m)$ cells. Verifying that all $n$ points of $P$ fall into the induced cell takes $O(nm)$ time. 
	Thus, the algorithm's total runtime is  
	$O(mn+k^2n)$. 
\end{proof}
%\section{Open Problems}
%We gave upper as well as lower bounds for the realizability problem for the continuous and the discrete \fd\ for curves in 1D and in 2D. Hence we now have a good understanding of the problem in these settings. % (except for algorithms for the discrete \fd\ and curves in 2D). 
%We now have a good understanding of the realizability problem in the settings studied. However, other settings remain open, in particular, when less information is given in the free space diagram or matrix, e.g. only some cells or only cell boundaries are given. %, or cells are not parameterized by lengths. 
%Beyond the settings we studied, some still remain open. 
%
%For instance, we would like to decide realizability when less information is given, e.g., %$\eps$ is variable, 
%only some cells are given, only cell boundaries are given, or cells are not parameterized by lengths. 
% -- perhaps even approximately. 
%Or perhaps only some of the cells are given, or these are not parameterized by lengths. %, and hence we do not know the segment lengths. 
%Furthermore, realizability by curves in  dimensions higher than 2 remains largely open. 
%\todo{reviewer comment:  many interesting questions remain open, and to me the most intriguing one is the question of whether a polynomial time algorithm can also be achieved in the case when the diagram is more "sparse", i.e., when there are many empty (or full) cells (while the algorithms in the paper only work for cases when there are many "partially full" cells).}

%\input{discussion}

\bibliography{references}

\appendix

\section{Omitted Proof from Section~\ref{sec:realdef}}\label{app:2Dlemma}

%%% lemma repeated from main text
\realpartially*
%\begin{lemma}\label{lem:realpartiallyfull}
%	Given a partially full free space cell, the distance $\eps$ and four points on the boundary of the ellipse within the cell, none of which are mirror images of the another one with respect to the ellipse's major and minor axes, we can compute the corresponding segments' relative placement.
%\end{lemma}

%%% old lemma statement
%\begin{lemma}\label{lem:realpartiallyfull}
%	Given a partially full free space cell $C_{i,j}$, the distance $\eps$ and four points on the boundary of the ellipse within the cell, none of which are mirror images of the another one with respect to the ellipse's major and minor axes, we can compute the corresponding segments' relative placement.
%\end{lemma}

\begin{proof}
	Knowing that the angles of the ellipse axes are at $45^\circ$, we can compute the ellipse's boundary equation:
	\begin{equation*}%\label{formula:ellipse}
		\left(\frac{(x-p_i^0)+(y-q_j^0)}{a}\right)^2+\left(\frac{(y-q_j^0)+(p_i^0-x)}{b}\right)^2=2,
	\end{equation*}
	where $a$ and $b$ denote the lengths of the ellipse's semi-major and semi-minor axes, and $(p_i^0,q_j^0)$ denotes its center point.
	
	Since the lengths of the segments $s_i^P, s_j^Q$ are given with the input diagram, %or, more precisely, its grid,
	to determine a realizing pair of segments for a partially full cell, it remains to compute the angle $\alpha$ between the two segments and the distances of the segments to the intersection point:
	\[\alpha = \arcsin\left(\frac{\eps/2}{a/\sqrt{2}}\right)=\arcsin\left(\frac{\eps}{\sqrt{2}a}\right).\]
	
	We observe that the sign of the angle does not affect the free space: consider a pair of segments $s_i^P$ and $s_j^Q$ corresponding to a given free space cell $C_{i,j}$. 
	We fix the placement of $s_j^Q$ and compute intersection and enclosed angle $\alpha$ of the two segments' supporting lines. 
	Now, placing $s_i^P$ such that the enclosed angle equals $-\alpha$ results in the same free space component as placing it at an enclosed angle of $\alpha$. 
\end{proof}

%As we now know the lengths, their intersection point and the angle between $s_i^P$ and $s_j^Q$, we can place them in the plane such that their relative positions match the input diagram's information. 
%We call such a placement of segments \emph{valid}.
%Note that by ``computing the segments' relative placements'' we refer to determining the pairwise distances between the two segments' endpoints.
\figref{fig:realcelltoplacement} illustrates the placement of segments: the extremal points of the ellipse's boundary curve correspond to intersections of one curve with the boundary of the second segment's $\eps$-neighborhood (marked with a circle).
For a fixed position of one segment, say $s_j^P$, the other segment's position is fixed up to symmetry.
That is, the enclosed angle can be either positive or negative, both images of $s_i^Q$ mirrored at $s_j^P$ are equally valid. 

\begin{figure}[htbp]
	\centering
	\includegraphics[scale=1]{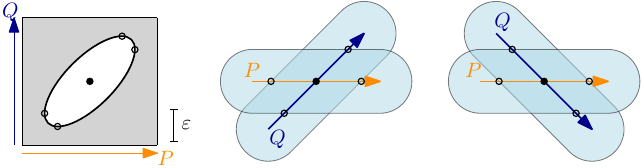}
	\caption{Input diagram consisting of one cell with marked extremal points and center, and both corresponding realizations. We fix $Q$ and place $P$ choosing a positive (left realization) or a negative (right realization) enclosed angle.}
	\label{fig:realcelltoplacement}
\end{figure}

\section{Omitted Proofs and Details from Section~\ref{sec:continuous2DHardness}}\label{app:ER-cont2D}

\begin{lemma}
	\label{lem:ER-containment-continuous}
	The problem of finding curves $P$ and $Q$ in $\Reals^2$ that realize an input diagram $D_\eps$ is in $\exists\Reals$.
\end{lemma}

\begin{proof}
	We reduce the problem to a system of real polynomial inequalities.
	The coordinate of vertices of $P$ and $Q$ are the variables.
	Each cell represents a constraint:
	The proof of %Lemma~\ref{lem:1Dpartfull} 
	Lemma~3.1~\cite{eurocg22} gives a quadratic equation relating the endpoints of the corresponding segments.
	Completely full cells require each segment to lie in the intersection of the $\eps$-neighborhood of the other segment's endpoints.
	Completely empty cells require that each segment lies outside of the other segment's $\eps$-neighborhood.
	All constraints can be expressed with a constant number of quadratic inequalities.
\end{proof}

%\begin{figure}[h]
%    \centering
%    \includegraphics[scale=.e6]{figures/full-example-f-s.pdf}
%    \caption{Example of our reduction from linkage realizability to \fsd realizability. (a) An input linkage $\mathcal{L}=(G,\Pi)$ and a subdivision vertex $v$ in a cycle of $G$. Rigid vertices are marked with gray angles. (b) Splitting $v$ transforms $G$ into a tree $T$. (c) The curves $P$ and $Q$ in $\Reals^2$ obtained from $T$. (d) The obtained \fsd.}
%    \label{fig:full-example-f-s}
%\end{figure}

\smallskip\noindent\textbf{Reduction description.}
Here, we detail the overview shown in Section~\ref{sec:continuous2DHardness}.
Refer to Figure~\ref{fig:full-example-f-s}.
Note that every edge in $Q$ corresponds to an edge of $T$. Every edge of $T$ has two correspondents on $Q$. 
Every edge in $P$ either corresponds to an edge of $T$ or is part of an angle gadget. 
A free space cell is empty unless it corresponds to a pair of edges of $P$ and $Q$ that (i) correspond to the same edge, (ii) correspond to adjacent edges and their shared vertex is rigid, (iii) one is part of an angle gadget and the other corresponds to one of the edges in the angle represented by the gadget, or (iv) correspond to a pair of edges in $T$ that were the result of a subdivision of an edge of $G$.
In case (i), the corresponding cell has a $\pm 45^\circ$ slab of width $\sqrt{2}$ depending on the direction of the traversal of the edge.
In case (ii), if the rigid acute angle between the corresponding edges of $G$ is $90^\circ$, the corresponding cell has a quarter of a unit disk centered at the corner that corresponds to the rigid vertex.
Else, the rigid angle is $180^\circ$ the corresponding cell has a right isosceles triangle with two edges of length $1$ and whose vertex incident to the right angle is at the corner of the cell that corresponds to the rigid vertex.
In case (iii), as explained in the description of the angle gadget, the cells corresponding to long edges are empty. 
The cells between a short edge of $P$ and an edge of $Q$ incident to the angle will contain half of a unit disk so that the distance between the half-disk and the corner of the cell that correspond to the nonrigid vertex is $1$. 
Finally, case (iv) is the same as case (ii) with a rigid angle of $180^\circ$.

\begin{lemma}\label{lem:rigid}
	Given a realization of $P$ and $Q$, the subcurves that correspond to a rigid subgraph $H$ exactly cover a rigid transformation of $H$. 
\end{lemma}
\begin{proof}
	By construction, every pair of edges from $P$ and $Q$ that either correspond to the same edge of $T$, or to two adjacent edges of $T$ that in turn correspond to edges in the same rigid subgraph $H$, define a partially full cell. Then, the claim follows by Lemma~\ref{lem:realpartiallyfull}.
\end{proof}

\section{Omitted Details from Section \ref{sec:continuous1D}}
%NP-hardness for continuous Fr\'echet distance and Curves in 1D}
\label{app:NPh1D}

%\subsection{NP-completeness of Realizability for Curves in 1D \label{sec:continuous1DHardness}}
%\label{sec:hard}
%We prove that realizability in 1D is weakly NP-hard by reducing from the \textsc{Partition} problem. 
%\carola{Hardness or completeness?}\leonie{You're right, it should be completeness. I changed that.}

We will show NP-hardness by a reduction from the partition problem. 
\subparagraph{Partition problem.} Given a set of positive integers $A = \{a_1, \ldots , a_n\}$, decide whether there exist two sets $A_1$, $A_2$, such that $\sum\limits_{\substack{a_i \in A_1}} a_i = \sum\limits_{\substack{a_j \in A_2}} a_j$, where $A_1 \cap A_2 = \emptyset$ and $A_1 \cup A_2 = A$.

\onedhardness*

\begin{proof}
Given partition instance $\{a_1, \ldots , a_n\}$, we set $S=\sum_{i=1}^n a_i$. 
We construct an input diagram of size $(n+2) \times 1$, see  \figref{fig:reduction}, and simplify our notation of cells $C_{i1}=C_i$. 
Realizing this constructed diagram through curves $P$, defined by endpoints $p_0, \ldots p_{n+2}$, and $Q$, consisting of a single segment $s^Q=\overline{q_0q_1}$, has to correspond to partitioning our integers into two sets of equal ``weight'' $S/2$.
The constructed diagram is a strip of height $\lvert s^Q \rvert$, which we set to~1. 
The total width of our diagram is set to $2(1+S)+S$:
Each cell width corresponds to the length of a segment in $P$, and we choose the length of a segment $s^P_i$ to be $\vert a_{i-1}\vert$, for $i=2, \ldots, n+1$. 
Segments $s^P_1$ and $s^P_{n+2}$ both have length $1+\sum_{i=0}^n\vert a_i\vert$, and we set $\eps=1$.

\begin{figure}[ht]
	\centering
	\includegraphics[scale=0.8]{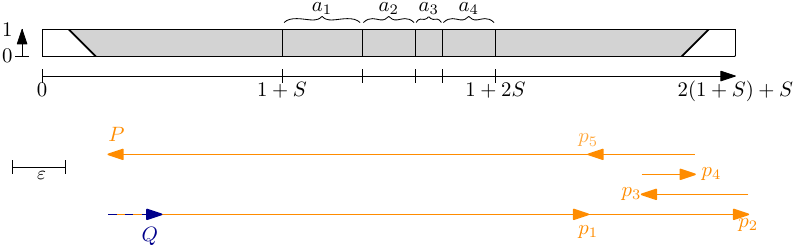}
	\caption{Reduction for partition instance $\{a_1, \dots , a_4\}=\{3, 2, 1, 2\}$. The corresponding segments are scaled by factor 2 and placed parallel instead of on top of each other to increase readability.}
	\label{fig:reduction}
\end{figure}

Now, the first and last cell of our constructed diagram are set to be partially full, the bottom left and top right corner being contained in white space.
All intermediate cells $C_2, \ldots C_{n+1}$ are empty.
For fixed position of $s^Q$, we construct the first and last cell such that segments $s^P_1$ and $s^P_{n+2}$ need to be aligned with $s^Q$; namely, we force $q_0=p_0=p_{n+2}$.
Consequently, $p_1=p_{n+1}$. 
For the remaining segments $s^P_2, \ldots , s^P_{n+1}$, we compare their orientation with placing the corresponding integer in either of the two sets $A_1, A_2$.
Starting with $s^P_2$, we can decide to place it on top of its predecessor, such that $\Vert p_2 - p_0\Vert = \Vert p_1 - p_0 \Vert-a_1$, or facing the same direction, such that $\Vert p_2 - p_0\Vert = \Vert p_1 - p_0 \Vert+a_1$.
If we choose the first option, we say $s^Q_1$ is \emph{oriented to the left}, otherwise it is \emph{oriented to the right}. 

Assume we are given an %solvable 
instance of the partition problem. 
The constructed strip is realizable iff we can place the curve $P$ such that points $p_1$ and $p_{n+1}$ coincide. 
This is the case iff the total length of segments oriented to the right equals the total length of segments oriented to the left.
Thus, the information on orientations of $s^P_2, \ldots s^P_{n+1}$ directly encodes a partition of integers $a_1, \ldots, a_n$ into sets $A_1$ and $A_2$.
We conclude
\begin{equation*}
	\sum_{a_i \in A_1}a_i =\sum_{a_j \in A_2}a_j =\frac{S}{2} \hspace{0.3cm} \Longleftrightarrow \hspace{0.3cm} p_{1}=p_{n+1}.
\end{equation*}
This proves NP-hardness.
The problem lies in NP since the realization of a polygonal curve  with $n$ segments can be described by a bit sequence of length $n-1$, specifying for each vertex whether the incident segments have the same orientation.
%
%denoting whether the shared endpoint of two segments is a fold, i.e., there is a change of orientations between two consecutive segments.
\end{proof}

\noindent {\bf FPT-algorithm for continuous curves in 1D.} In this case the free space diagram has limited ``configurations''. 
Cells are still empty, full or partially full cells, but now free space ellipses degenerate to slabs and the white space is bounded by  parallel line segments tilted by $\pm 45^\circ$, see~\cite{Buchin, Rote} and Figure~\ref{fig:partiallycell}. 
In this case, we can infer relative placements of segments already from one point on the free space boundary and cell boundary. 

\begin{figure}[tb]
	\centering
	\includegraphics[scale=0.65]{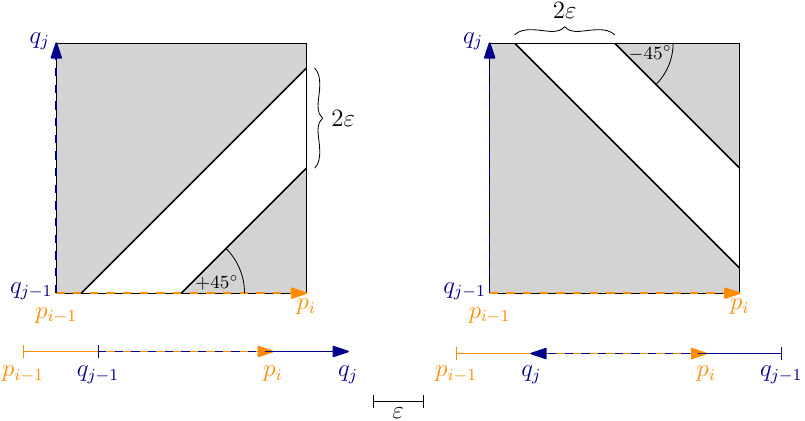}
	\caption{%Given a free space in white and its complement in grey,
		Partially full cells $C_{ij}$ of segments $s^Q_j$ and $s^P_i$ oriented in the same or  opposite direction.}
	\label{fig:partiallycell}
\end{figure}

%\onedcells*

\begin{lemma}
	\label{lem:1Dpartfull}
	For each partially full cell $C_{ij}$ corresponding to a pair of segments $s^P_i, s^Q_j$ in 1D, there exists at least one point $x$ on the intersection of the free space boundary and the cell boundary. 
	Its position fixes the distance of two endpoints $(p,q) \subset \{p_i, p_{i+1}\}\times \{q_j,q_{j+1}\}$, and thus allows to fully determine the relative positions of $s^P_i$ and $s^Q_j$. 
\end{lemma}

\begin{proof}
	As both segments are placed on the real line, at least one endpoint lies with $\eps$ distance of a point of the other segment. 
	W.l.o.g., we fix the position of $s^P_i$ such that $p_{i-1}=0$, $p_i=\vert s_i^P \vert$, where $\vert s_i^P \vert$ denotes the segment's length.
	Let $x$ lie on the left boundary of $C_{ij}$, which corresponds to $p_{i-1}\times s_j^Q$, and we call $d(x)$ the distance between $x$ and the bottom left corner $p_{i-1} \times q_{j-1}$, see \figref{fig:partiallyfull}.
	For $d(x)=0$, we have that $\Vert p_{i-1}- q_{j-1} \Vert = \eps$, for $d(x)=\vert s^Q_j\vert$ it holds that $\Vert p_{i-1}-q_{j}\Vert = \eps$.
	In both cases, the orientation of $s^Q_j$ depends on whether $C_{ij}$ contains a free space region.
	Iff this is the case both segments have the same orientation.
	If $0 < d(x) < \vert s^Q_j \vert$ there could be one such intersection point, or two at distance $2\eps$, see \figref{fig:partiallycell}.
	Assuming $x$ denotes the lower one, i.e., the interval between bottom left corner and $x$ is not contained in free space, it holds that $\Vert p_{i-1}-q_{j-1}\Vert = d(x) + \eps$, because the point $q_x \in s^Q_j$ at distance $d(x)$ from $q_{j-1}$ has distance exactly $\eps$ to $p_{i-1}$.
	The mirrored case holds for $x$ denoting the upper intersection point and endpoints $p_{i-1}, q_j$.
	The orientation of $s^Q_j$ depends on the angle of the free space boundary within the cell; for $+45^\circ$, both segments $s^P_i, s^Q_j$ face in the same direction.
\end{proof}

\begin{figure}[htb]
	\centering
	\includegraphics[scale=0.65]{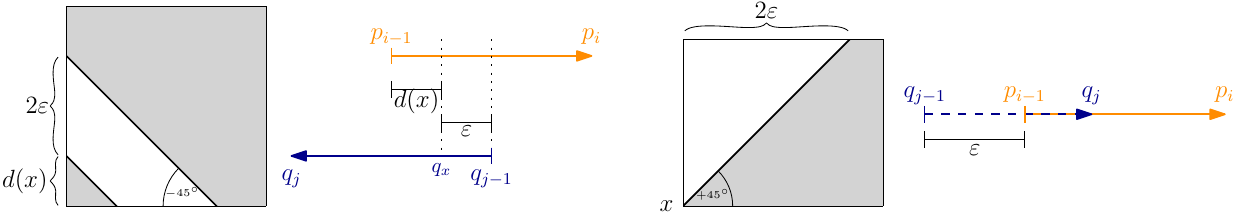}
	\caption{Placing $s_j^Q$ for fixed $s_i^P$ to realize a corresponding partially full cell.}
	\label{fig:partiallyfull}
\end{figure}

Next we observe a symmetry in the free space, see \figref{fig:mirror}.

\begin{observation}
	\label{obs:folding}
	Free space diagrams of curves in 1D are in some sense symmetrical, as described in~\cite{Buchin}: 
	Consider a point $p_i \in P$ at which the curve folds, and some point $q \in Q$. 
	We choose points $p \in s^P_i, p' \in s^P_{i+1}$ that are equidistant to $p_i$, so $p=p' \in \Reals$.
	Now $\| p-q \| \leq \eps$ holds iff $\| p'-q \| \leq \eps$. 
	Thus, the strip to the right of the grid line $p_i \times Q$ is a reflection of the strip to the left of that line.
	For consecutive partially full cells, this implies that a curve folds at the common endpoint $p_i$ iff the incident portions of free space have alternating slopes. 
\end{observation}

We now borrow some definitions from computational origami, giving intuitive descriptions.
We refer to~\cite{Demaine} for formal definitions.
The \emph{crease pattern} $C(D_\eps)$ of a diagram $D_\eps$ is the crease pattern obtained by considering grid lines that correspond to folding vertices as creases. 
The \emph{folded state} of a crease pattern $C(D_\eps)$ is a continuous function that maps each face isometrically and reflects adjacent faces.
From \obsref{obs:folding}, we conclude 

\begin{corollary}
	\label{cor:foldability}
	A given diagram $D_\eps$ is realizable iff there is an assignment of the grid lines to {\em \{\texttt{fold}, \texttt{straight}\}}, such that overlapping white space aligns in the folded state $C(D_\eps)$.
\end{corollary}

\todo{reviewer comment: I am not familiar with computational origami. More figures would help in understanding the notions of "folded state", "crease pattern", "end fold", "crimp", etc. Also, add more details in the description of Figure \ref{fig:mirror}.}

\begin{figure}[ht]
	\centering
	\includegraphics[scale=.75]{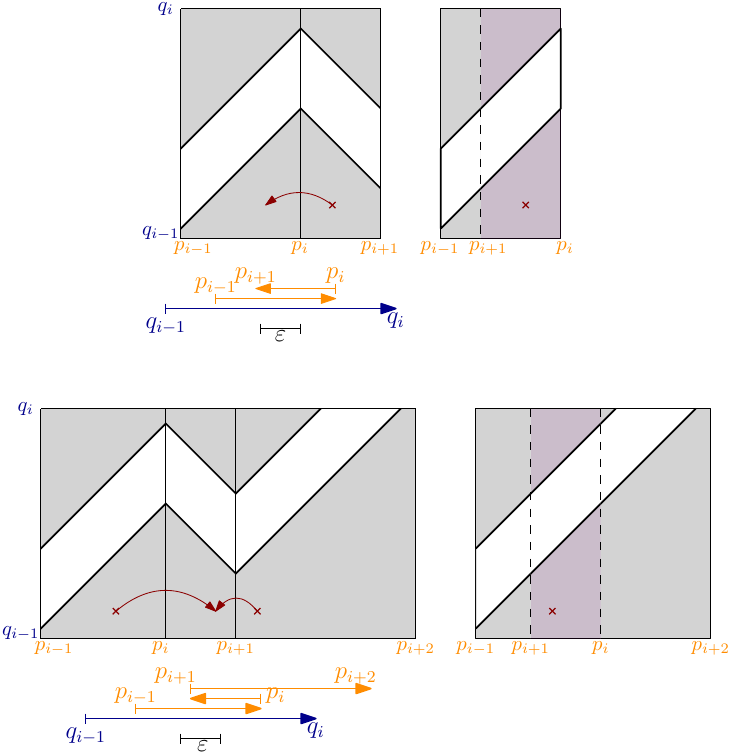}
	\caption{An end fold and a crimp of two diagrams along the grid lines of folding vertices.}
	\label{fig:mirror}
\end{figure}

%We also note that the $\eps$-neighborhoods of curves $P$ and $Q$, i.e., the union of their segments' $\eps$-neighborhoods, are connected. 
%Thus their overlap is connected, as well. 
%In combination with out previous observation, this yields that if strips contain empty cells, there must be at least two consecutive empty cells in between partially full cells. 
%However, the number of consecutive empty cells is not necessarily even, as the curves may not fold at all endpoints.

%\todo{Add partial diagrams? Add Theorem stating total runtime.}

%\leonie{Note that the diagonal case can only happen with full cells!}

%\subparagraph{Algorithm.} %Draft.
We use \corref{cor:foldability} to obtain an algorithm for 1D continuous realizability. 
The algorithm runs in exponential time $O(mn2^k)$ for general inputs. Here, $k$ is the number of rows of $D_\eps$ that do not intersect the boundary of the white space, i.e., the number of vertical or horizontal strip ``gaps'' (completely gray or completely white).
Intuitively, gaps correspond to segments of one curve that are either too far or too close to the other curve, and do not provide direct information about curve  placement.
%In particular, 
If there is a white gap in $D_\eps$ (a row of only full cells) then the diameter of one of the curves is smaller than $2\eps$.
For inputs where $k\in O(\log (mn))$, the algorithm thus runs in polynomial time in the size of the diagram.

For a given assignment in \corref{cor:foldability}, one could naively check if the corresponding folded state is valid by checking the pairwise overlap of free-space cells in $O((mn)^2)$ time. 
\todo{reviewer comment: Given one of the $2^k$ folding patterns, I believe that if we check its validity from left to right, then we always just have a single slab for each row that we have to check consistency with or extend further after folding, which clearly is an $O(nm)$ algorithm. Is there something that I am missing about the difficulty? If so, consider pointing it out more clearly}
Instead, we achieve a better runtime by incrementally checking for consistency using some theory from origami.
Our algorithm is inspired by the algorithm for \emph{simple-foldability in 1D}~\cite{origami},
%by Arkin et~al., 
which asks whether a 1D crease pattern can be folded through a sequence of simple folds ($\pm 180^\circ$ rotations of a portion of the paper).
Their linear-time greedy algorithm constructively decides whether a given crease pattern admits a sequence of simple folds.
Although it works with crease patterns where each crease has a \emph{mountain/valley} assignment,
%(whether the crease should be folded by $+180^\circ$ or $-180^\circ$), 
for our problem this assignment is irrelevant.
Without a mountain/valley assignment, every crease pattern is simple-foldable.
We use their algorithm for the properties that it maintains during a linear-sized sequence of operations (simple folds).
The algorithm identifies one of two sufficient operations that can be greedily applied:
(i) an \emph{end fold} which is a simple fold applied to either the first or last crease; and (ii) a \emph{crimp} which folds through two consecutive creases.
They show that these operations are \emph{safe} if 
\begin{enumerate}
	\item [$(\star)$] the resulting folded state after each operation does not cause a portion of the paper without creases to overlap with a not yet folded crease.
\end{enumerate}
After applying the operation, they reduce the problem to a smaller crease pattern obtained by ``gluing'' the overlapping layers.
They show that if the crease pattern is not trivial (no creases), there is always a safe operation, which they can find in $O(1)$ time with $O(n)$-time preprocessing.
We call any operation (end fold or crimp) \emph{valid} if the resulting folded state aligns the white space patterns, see \figref{fig:mirror}.
We use the above mentioned algorithm to efficiently check if a given crease pattern induces a realizable input space diagram. %\maike{todo: add a sentence why this is efficient}
%\maike{Couldn't we also simply use crimp folds from left to right (i.e. an alternating M-V sequence), which are also safe?}
%\hugo{Crimp can only be done in local minima. If two crimps are possible any one of the two can be done so you can indeed do crimps from left to right. But the first crimp could be the 20th and 21st creases and not creases 1 and 2. 
	%E.g.: the lengths of the segments are 5,4,3,2,1,2,3,4,5. Then only the creases around the segment of length 1 are crimpable. We could also compute the map where each layer of paper will end up after we fold everything (Which is maybe what you're suggesting). The problem here is that if we check this way each cell in the arrangement have to be checked up to $n$ times because it is the overlap of up to $n$ layers. This makes this method quadratic.}
%\maike{thanks! yes, I was thinking of the second option, but see the problem with it now}

% \begin{theorem}
	%     \label{thm:alg-fpt-continuous1d}
	%     Given an $m\times n$ diagram $D_\eps$, we can find 1-dimensional curves $P$ and $Q$ such that $D_\eps=D_\eps(P,Q)$, if they exist, in $O(mn2^k)$ time, where $k$ is the total number of (vertical and horizontal) grid lines of $D_\eps$ that do not intersect the white space. 
	% \end{theorem}

\algfptcontinuous*

\begin{proof}
	We describe a constructive algorithm.
	The main idea is to test every possible corresponding crease pattern that is compatible to $D_\eps$.
	Note that if a grid line intersects the white space, we can determine if it corresponds to a folding vertex or not as follows. \todo{reviewer comment: I think you miss the case where the whole grid line is white. It is obvious, but you should still handle it.}
	If two adjacent cells align the white space after a reflection through the grid line, then the grid line should be assigned to \texttt{fold}; otherwise, we assign it to \texttt{straight}.
	If there is an inconsistency of the assignment given by different pairs of adjacent cells incident to a grid line, the instance is not realizable, and we return ``no''.
	For the remaining $k$ grid lines that do not intersect the white space, we try all possible assignments to $\{\texttt{fold},\texttt{straight}\}$.
	We delete all grid lines that are assigned \texttt{straight}, merging pairs of adjacent cells through the deleted line.
	That defines a crease pattern and allows us to check realizability using \corref{cor:foldability}.
	
	%It remains to show how we efficiently check whether a given crease pattern satisfies \corref{cor:foldability}.
	%Note that 
	Now, we could pay an extra linear factor using brute force to check for every pair of overlapping cells of $D_\eps$ whether their white space aligns, leading to an $O(m^2n^2)$ time algorithm.
	Instead, we use the simple-foldability algorithm~\cite{origami} to obtain an $O(mn)$-time algorithm as follows.
	We fold one dimension at a time.
	W.l.o.g., we focus on the horizontal dimension, corresponding to $P$, which contains $O(n)$ creases.
	Identify a safe operation and apply the corresponding fold(s). 
	Note that each operation causes at most three layers to overlap.
	Recall that we merge these layers into a single layer by ``gluing'', and apply induction.
	It suffices to check the alignment of the white space in the overlapping layers.
	By property~$(\star)$, the number of cells in these layers is $O(m)$, so the check can be performed in $O(m)$ time.
	In future operations, if the white space of a merged cell and the white space of another cell align, the white space of all original overlapping cells align since alignment is transitive.
	This proves the induction step.
	After $O(m)$ operations, we obtain a single segment in the horizontal dimension, and we can apply the same algorithm for the vertical dimension.
	The runtime of the check step is then $O(mn)$, proving our claim.
	
	%In the special cases 
	When the input diagram is completely empty or completely full, \corref{cor:foldability} is trivially satisfied.
	If the diagram is empty, we place the curves sufficiently far apart from each other (farther than $\eps$).
	If it is full, we compute the intersection of the $\eps$-neighborhoods of every segment of one curve, and check whether we can place the other curve in this intersection. 
	%Recall that 
	With a fixed crease pattern, the diameters of the curves are deterministically defined.
\end{proof}

\section{Omitted Details from Section \ref{sec:continuous1D}}
%NP-hardness for continuous Fr\'echet distance and Curves in 1D}
\label{app:ppt1D}

\subsection{Fixed Boundary Subproblems}
\label{app:fixed-boundary}

%Full details are given in Appendix~\ref{a}
In this case, we know the relative position of the boundary points of the uncertainty region.
That includes left and right uncertainty regions, and middle uncertainty regions when there is a single non-singleton component in $G$.
Without loss of generality, we assume that the left boundary points of the uncertainty region is $0$. % and the region lies in the positive axis \leo{I find this formulation a bit odd, maybe "and all coordinates are positive"?}.
We describe the DP for a region $[0,r]$.
The modification for $r=\infty$ is straightforward.
Let $Q'=(q_i,\ldots, q_{j-1},q_j)$ be a maximal subcurve of $Q$ in the uncertainty region $[0,r]$.
The segments $q_iq_{i+1}$ and/or $q_{j-1}q_j$ might come with a \emph{boundary constraint}, i.e.,  whether $q_i$ (resp., $q_j$) is mapped to $0$ or $r$.
Such conditions depend on whether the segments $q_{i-1}q_i$ and $q_jq_{j+1}$ exist and in which certainty region they are.
We assume that at least one of them is subject to a boundary constraint, or else $D_\eps$ is completely empty and the problem is trivial: embed the two curves far apart.
Without loss of generality, assume that $q_j$ is constrained to be at~$0$.
We define the subproblem $R(k, s)$, $i\le k \le j$ as \texttt{true} if the subcurve $(q_k,\ldots,q_j)$ can be embedded in $[0,r]$ while fixing the position of $q_k$ to the point $s$, where $0\le s\le r$.
The base case is $R(j, 0)=\texttt{true}$ and $R(j, s)=\texttt{false}$ for $s\neq 0$.
The recursive case below tries both possible orientations for $q_kq_{k+1}$ (pointing towards the right or left) and checks if that causes $Q'$ to go outside of $[0,r]$.
\begin{align}
\label{eq:DP}
\begin{split}
	R(k, s) = &(R(k+1, s+\|q_k q_{k+1}\|)\wedge(s+\|q_k q_{k+1}\| <r)) \ \vee \\
	&(R(k+1, s-\|q_k q_{k+1}\|)\wedge(0< s-\|q_k q_{k+1}\|))).       
\end{split}
\end{align}

We call a subcurve $Q'$ \emph{realizable} if $R(1, s)=\texttt{true}$ for some $s$ and $q_1$ has no boundary constraint, or $R(1, 0)=\texttt{true}$ if %$q_1$ is constrained to be at $0$
$q_1=0$, or $R(1, r)=\texttt{true}$ if %$q_1$ is constrained to be at $r$.
$q_1=r$.

%\fixBoundary*
%\begin{proof}
%    Since $k\in\{1,\ldots,n'\}$ and $s\in\{0,\ldots,r\}$ there are at most $O(n'r)$ subproblems.
%    We can also upper-bound $s$ by $n'W$ since this is the maximum length of the image of $Q'$ (which is necessary in the case when $r=\infty$).
%    Each subproblem can be computed in $O(1)$ time.
%    Thus the total runtime is $O(n'\cdot\min(r, n'W))$.
%\end{proof}

\subsection{Variable Boundary Subproblems}
\label{app:var-boundary}

We cannot fix both boundaries of the middle uncertainty region if both curves have segments of type (\ref{type:close}) (and thus have middle uncertainty regions).
The span of both curves depend on the sizes of their uncertainty regions, which in turn impose constraints on the size of the uncertainty region of the other curve and, thus, we cannot fix the boundary of the DPs for these regions.
We further divide into three cases based on boundary constraints.
In the first case, $D_\eps$ is completely full, meaning that both curves lie entirely in the middle uncertainty region of the other (every segment is of type (\ref{type:close})).
Thus, none of the segments has boundary constraints.
In the second case, $G$ has at most one non-singleton component. 
Then every maximal subcurve in the middle uncertainty region has at least one end point at the boundary point between the uncertainty region and the certainty region containing the segments in the non-singleton component.
In the third case, $G$ has two non-singleton components, so there exists at least one maximal subcurve in the middle uncertainty region that has endpoints at the two boundary points of the uncertainty region (spanning the entire region) since the curves are connected and contain segments in both their certainty regions.

We solve all cases with a slight variation of the DP for the fixed boundary case where we also guess  $\alpha$, the length of the image of the curve in the uncertainty region.
The subproblem $R(k, s, \alpha)$ where $k\in\{i,\ldots,j\}$ and $s\in \{0,\ldots, \alpha\}$ is then defined as before, and $\alpha$ is an integer in $\{1,\ldots,2\eps\}$.
Note that we can upper-bound $\alpha$ by $2\eps$ since middle uncertainty regions are defined by the intersection of two $\eps$ radius disks.
The recursive cases are computed with Equation~\ref{eq:DP} replacing $r$ with $\alpha$.
Now the base cases will depend on the type of boundary constraint.
For problems with no boundary constraint, we simply set $R(j, .,.)=\texttt{true}$.
For problems with at least one boundary constraint, we assume without loss of generality that $p_j$ is constrained to be at $0$.
Thus $R(j, 0,.)=\texttt{true}$ and $R(j, s,.)=\texttt{false}$ for $s\neq 0$.

As the lengths of the middle uncertainty regions depend on each other, we need to check if a compatible solution exists among all solutions of the DPs.
Let $r_P$ (resp., $r_Q$) be the length of the middle uncertainty region of $P$ (resp., $Q$) in a positive solution.
For a subcurve $Q'=(q_i,\ldots,q_j)$ with no boundary constraint at $q_i$ (resp., when $q_i=0$), %$q_i$ is constrained to be at $0$), 
we call a solution $R(i, .,\alpha)=\texttt{true}$ (resp., $R(i, 0,\alpha)=\texttt{true}$) \emph{compatible} with $r_Q$ if $\alpha\le r_Q$.
If $q_i$ is constrained to be at $\alpha$, then we call a solution $R(i, \alpha,\alpha)=\texttt{true}$ \emph{compatible} with $r_Q$ if $\alpha = r_Q$.

\section{Omitted Proofs from Section \ref{sec:discrete2DHardness}}
\label{app:discrete}

\ERdiscrete*

% \begin{figure}[ht]
	%     \centering
	%     \includegraphics[width=\textwidth]{figures/er-discrete.pdf}
	%     \caption{Reduction from $S=\{(-,-,-),(+,+,+), (-,-,+), (-,+,+), (-,+,-), (+,+-), $ $(+,-,-)\}$. (a) Transforming a solution to $M_\eps$ into a line arrangement that realizes $S$.
		%     (b) Transforming a solution to $S$ into curves $P$ and $Q$ that realize $M_\eps$. Here squares represent points of $P$ and circles represent points of $Q$.}
	%     \label{fig:er-discrete}
	% \end{figure}

\begin{proof}
	Containment in $\exists\Reals$ can be proven by a straightforward reduction to $\exists\Reals$ similar to the proof of Lemma~\ref{lem:ER-containment-continuous}.
	We now focus on the reduction defined above.
	It is clear that it runs in polynomial time.
	Assume that there exists a pair of curves $P$ and $Q$ that realizes $M_{\eps}$.
	Refer to Figure~\ref{fig:er-discrete}(a).
	We use the labels of point of $P$ defined in the reduction and assume $Q=(q_1,\ldots, q_{|S|})$.
	Recall that, informally, points $q_1$ and $q_2$ represent vectors $\textbf{v}_1=(-,\ldots,-)$ and $\textbf{v}_2=(+,\ldots,+)$, respectively.
	Rotate the solution in order to make the vector $\overrightarrow{q_1q_2}$ vertical and pointing upwards.
	We build a hyperplane arrangement as follows.
	For each $i\in [n]$, create a hyperplane $\ell_i$ bisecting the segment $a_i b_i$.
	Now, we argue that $q_j$, $j\in\{1,\ldots,|S|\}$ is in a cell in the produced arrangement with description $\textbf{v}_j$.
	Let $C_1$ and $C_2$ be the cells in the arrangements of circles of radius $\eps$ containing $q_1$ and $q_2$, respectively.
	By definition, if $\textbf{v}_j[i]=+$, then $q_j$ must be within $\eps$ distance from $a_i$ and farther than $\eps$ from $b_i$, that is $q_j\in \mathcal{B}_\eps(a_i)\setminus\mathcal{B}_\eps(b_i)$.
	Thus, $C_1= (\bigcap_{i=1}^n \mathcal{B}_\eps(b_i)\setminus \bigcup_{i=1}^n \mathcal{B}_\eps(a_i))$ and $C_2= (\bigcap_{i=1}^n \mathcal{B}_\eps(a_i)\setminus \bigcup_{i=1}^n \mathcal{B}_\eps(b_i))$.
	Note that every hyperplane $\ell_i$ must separate $C_1$ and $C_2$ by definition.
	Thus every $\ell_i$ intersects the line segment $\overline{q_1q_2}$.
	We focus on a specific hyperplane $\ell_i$.
	Without loss of generality assume $\textbf{v}_j[i]=+$.
	Then, $\mathcal{B}_\eps(a_i)\setminus\mathcal{B}_\eps(b_i)$ is above $\ell_i$ and so is $q_j$.
	Therefore, the produced hyperplane arrangement realizes $S$.
	
	Now assume that there exists a hyperplane arrangement realizing $S$.
	Refer to Figure~\ref{fig:er-discrete}(b).
	For each cell in the arrangement described by $\textbf{v}_j$, choose a point $q_j$ in the interior of the cell.
	As before, every hyperplane intersects the line segment $\overline{q_1q_2}$, since $q_1$ is below all the hyperplanes and $q_2$ is above.
	Let $t_i$ be the intersection of   $\ell_i$ and $\overline{q_1q_2}$.
	Define the balls $\mathcal{B}_r(w_{i,r}^+)$ and $\mathcal{B}_r(w_{i,r}^-)$ respectively above and below $\ell_i$, tangent to $\ell_i$ at $t_i$.
	Note that $\mathcal{B}_r(w_{i,r}^+)$ (resp., $\mathcal{B}_r(w_{i,r}^-)$) equals the upper (resp., lower) halfspace of $\ell_i$ when $r\rightarrow \infty$.
	Thus, for some sufficiently large $r$, $\mathcal{B}_r(w_{i,r}^+)$ contains all points $q_j$ above $\ell_i$ and $\mathcal{B}_r(w_{i,r}^-)$ contains all points $q_j$ below $\ell_i$.
	Let $r^*$ be a sufficiently large $r$ such that the previous statement is true for all $i\in [n]$.
	Scale the entire construction to make $r^*=\eps$.
	Then, we can construct $P$ by making $a_i=w_{i,\eps}^+$ and $b_i=w_{i,\eps}^-$.
	By construction, each $q_j$ is contained in the appropriate cell of the arrangement of circles of radius $\eps$ centered at points of $P$.
	Thus, the constructed $P$ and $Q$ realize $M_\eps$.
\end{proof}

\section{Omitted Details from Section \ref{sec:discrete1Dpoly}}
\label{app:discretePoly}
\vspace{-1mm}

%For each row ($p \in P$), return `YES' if the existence of the interval containing the correct labeling from the unit intervals is identified, and return `NO', otherwise. 
The full description of realizability of discrete curves in $\Reals^1$ is presented below:

\noindent$\textsc{DiscreteRealizabilityAlgo}(M_\eps):$
%\carola{Goal: We want to take the partial order that we get from the initial BFS, and refine it to be closer to a global order. In that process we will order most indistinguishable vertices; a few will remain indistinguishable, we call those "equivalence classes". }

%\carola{Terminology: Indistinguishable: BFS steps 2-4: still have a couple of vertices with same neighborhood set ("indistinguishable"). And our goal is to disambiguate them (meaning, give them an order -- could also say we want to "order" them) so that we arrive at something that's almost a global order (a few might still stay indistinguishable, we call those equivalence classes). }

{\bf Step (1):}  Construct an abstract unit-interval graph $G$ from $M_{\eps}$ whose rows are vertices in $P$ and columns are vertices in $Q$: The vertex set of $G$ represents the columns of $M_{\eps}$, i.e., the interval of length $2\eps$ centered at a point $q\in Q$. For every row ($p \in P$), the edge set of $G$ is defined by including a clique between the columns (vertices in $Q$) that are filled with $1$ in that row. 
We store the edges of $G$ in an adjacency matrix. 
In the following we assume $G$ connected, or else, we treat each component separately.
%\carola{ Then $|E|=O(m^2)$ storage but we need $O(k^2n)$ time to add all edges. The number of edges $|E|=min(m^2, k^2n$). Or, we could live in **adjacency lists** world, add in duplicate edges, and then sort them to get rid of duplicates -> Runtime (edge generation + sorting) $k^2n + k^2n\log(k^2n)$ or with **counting sort**: $k^2n + k^2n + m$ and since $m<n$ -> $O(k^2n)$.}
\label{step:uig_const}
{\bf Step (2):}
Choose a left anchor as follows. 
As in~\cite{CORNEIL}, we get a set of candidates for left anchor by running a BFS search on $G$ from an arbitrary vertex.
The set contains the vertices with minimum degree in the deepest level of the BFS.
The set of candidates must be in at most two equivalence classes, or else $G$ is not an interval graph.
We augment $G$ by adding a new vertex $v_0$ connected to each vertex in one of the equivalence classes of candidates. 
We choose $v_0$ as a left anchor.
%Since all candidates for the left anchor are indistingushable, we can compute all such candidates by performing a BFS search on $G$. The ultimate left anchor is obtained by connecting a new vertex $v_0$ to the candidates. 
%by performing a BFS search on $G$ while trying all ties. 
%
%\carola{BFS takes time $O(|V|+|E|)=O(m + m^2)$ if in adjacency matrix form, or if we convert to adjacency lists then $O(m+min(m^2, k^2n)$}
%
{\bf Step (3):}  Perform a BFS on $G$ starting at $v_0$ to obtain a partial order of the vertices.  
{\bf Step (4):} Refine the partial order to get the global order by sorting the intervals at each level of the BFS under the criterion of $|\nex(N(v))|- |\pre(N(v))|$. 
We denote the current partial order as $D$ and use $<_D$, $>_D$ and $=_D$ to denote whether intervals appear in order, in reverse order, or are incomparable in $D$, respectively.
By Theorem 2.2 in~\cite{CORNEIL}, a pair of vertices $u$ and $v$ with $u=_D v$ are indistinguishable.
Thus a set of pairwise incomparable vertices in this partial order forms an equivalence class.
Add a vertex $v_f$ to the end of the order, connecting it to all vertices in the last equivalence class.
%This will leave a set of indistinguishable intervals whose $D(v)$ values are equal. We call these an {\em equivalence class $C$}. 
%
{\bf Step (5):} 
Further refine the partial order as follows. We refer to such refinement as $D'$.
For each row $r$ in $M_\eps$, let $I_r$ be the set of intervals with entry 1 in $r$, and let $C' = C\cap I_r\neq \emptyset$ where $C$ is an equivalence class.
If there are $i \in I_r\setminus C$ and $c'\in C'$ where $i<_D c'$ (resp., $i>_D c'$), make $c'<_{D'} c$ (resp., $c'>_{D'} c$) for all $c\in C\setminus C'$.
%
% \begin{enumerate}[(i)]
	% \item if $r$ contains intervals $I$ with entry 1, 
	%
	% \item if there is a proper subset $C' \subset C$ whose intervals have entry 1 in row $r$, 
	%
	% \item $D(i)<D(c)$ for all $i \in I\backslash C$ and $c \in C$, \end{enumerate}
%
% then all intervals in $C'$ occurs before $C \backslash C'$ and after $I \backslash C$. Update $D(c')$, for all $c' \in C'$, appropriately.
%
%(or intervals $v$ with entry 1 and of $D(v)>D(v')$ for $v' \in C$), then all intervals in $C'$ occurs before (resp. after) $C \backslash C'$.
% 
{\bf Step (6):} Extend the partial order defined by the BFS levels and $D'$ to a global order breaking ties arbitrarily.
Obtain the arrangement of unit intervals based on the global order and $G$.
This can be done as in Theorem 3.2 in~\cite{CORNEIL}.
%obtained order with respect to $D'(v)$ for all intervals $v \in G$ into a global order arbitrarily; those intervals $v$ that are indistinguishable (with identical $D(v)$) can be replaced. 
%
{\bf Step (7):} Verify whether the produced arrangement is compatible with $M_\eps$. 
Each row $r$ specifies the existence of a cell where exactly the intervals with entry 1 intersect.
If a cell specified by a row does not exist in the arrangement, return `NO'.
Otherwise, return `YES'.

\end{document}